\documentclass[12pt, draftclsnofoot,onecolumn]{IEEEtran}
\usepackage{amsmath,amsfonts}
\usepackage{amsmath,amssymb}
\usepackage{algorithmic}
\usepackage{algorithm}
\usepackage{array}
\usepackage{mathtools}
\usepackage{textcomp}
\usepackage{url}
\usepackage{cite}
\usepackage{hyperref}
\usepackage{verbatim}
\usepackage{graphicx}
\usepackage[shortlabels]{enumitem}
\usepackage[center]{caption}
\usepackage{dsfont}
\usepackage{xcolor}
\usepackage{pgf}
\usepackage{tikz}
\usetikzlibrary{arrows,automata}
\usepackage[latin1]{inputenc}

\usepackage[font=small]{caption}
\usepackage{subcaption}

\usepackage{amsthm}
\newtheorem{theorem}{Theorem}
\newtheorem{lemma}{Lemma}
\newtheorem{definition}{Definition}

\def\BibTeX{{\rm B\kern-.05em{\sc i\kern-.025em b}\kern-.08em
		T\kern-.1667em\lower.7ex\hbox{E}\kern-.125emX}}
\usepackage{balance}
\DeclareMathOperator*{\argmin}{arg\,min}

\begin{document}
	\setlength{\abovecaptionskip}{0pt}
	\setlength{\belowcaptionskip}{0pt}
	
	\title{State-Aware Timeliness in Energy Harvesting IoT Systems Monitoring a Markovian Source}
	\author{Erfan Delfani, George J. Stamatakis, and Nikolaos Pappas \thanks{E. Delfani and N. Pappas are with the Department of Computer and Information Science Link\"{o}ping University, Sweden, email: \{\texttt{erfan.delfani, nikolaos.pappas\}@liu.se}. G. Stamatakis is with the Institute of Computer Science, Foundation for Research and Technology - Hellas (FORTH), email: \texttt{gstam@ics.forth.gr}. A shorter version has been published in \cite{stamatakis2019control}.}}
	
	\maketitle
	
	\begin{abstract}
		In this study, we investigate the optimal transmission policies within an energy harvesting status update system, where the demand for status updates depends on the state of the source. The system monitors a two-state Markovian source that characterizes a stochastic process, which can be in either a \emph{normal} state or an \emph{alarm} state, with a higher demand for fresh updates when the source is in the alarm state. We propose a metric to capture the freshness of status updates for each state of the stochastic process by introducing two Age of Information (AoI) variables, extending the definition of AoI to account for the state changes of the stochastic process. We formulate the problem as a Markov Decision Process (MDP), utilizing a transition cost function that applies linear and non-linear penalties based on AoI and the state of the stochastic process. Through analytical investigation, we delve into the structure of the optimal transmission policy for the resulting MDP problem. Furthermore, we evaluate the derived policies via numerical results and demonstrate their effectiveness in reserving energy in anticipation of forthcoming alarm states.
	\end{abstract}

	\section{Introduction}
	Timely communication of status updates is critically essential for applications providing monitoring services in cyber-physical systems~\cite{yates2021aoi}. These applications form the foundation of the intelligent infrastructure enabled by the Internet of Things (IoT). Instances of such applications encompass, but are not restricted to, smart cities, intelligent factories and grids, advanced agriculture, parking and traffic control, e-Health, and environmental monitoring~\cite{yates2021aoi,abd2019role}. 
	
	A pivotal finding in the field indicated that metrics like throughput and delay do not adequately address the goal of timely status updating. In addressing this issue, the authors in~\cite{kaul2012real} introduced a novel metric known as the Age of Information (AoI).
	Since its introduction, the optimal determination of status update generation and transmission to minimize AoI metrics has garnered considerable attention from the research community \cite{SunTIT2017,sun2019sampling, stamatakis2020optimal}. The scope of AoI has been extended to encompass other metrics, including the Value of Information \cite{kosta2017age}, cost of update delay \cite{sun2018sampling}, Age of Synchronization \cite{zhong2018two}, non-linear AoI \cite{zheng2019closed}, Age of Incorrect Information \cite{maatouk2020age}, Version Age of Information \cite{yates2021age}, and Age of Actuation \cite{nikkhah2023age}.
	
	Another notable challenge in the field involves selecting a suitable energy source for remote sensors. With their finite lifespan, batteries pose the risk of high replacement costs, particularly when dealing with numerous sensors located in remote or inaccessible areas. To tackle this issue, energy harvesting (EH) technologies have been devised to provide the required power to remote sensors~\cite{akan2017internet}. Regardless of whether energy harvesting, batteries, or both are employed, the stored energy must be judiciously managed to ensure an adequate supply when most crucial.
	
	In\cite{lazy_timely}, the paper explores the optimization of transmitting updates from an EH source to a receiver, aiming to minimize the average age of updates. Similar studies can be found in \cite{ArafaGC2017,ArafaICC2018, WuYangTGCN18, farazi2018age,feng2018minimizing,arafa2019age,bacinoglu2019optimal,feng2021age}. 
    The paper \cite{abd2020reinforcement} explores a monitoring system where nodes are powered wirelessly and send updates to a central node to maintain data freshness. It aims to minimize the average AoI by optimizing energy transfer and update scheduling. Using deep reinforcement learning, the paper proposes an efficient solution and analyzes its properties. It also compares the optimal policy with one maximizing throughput and studies the impact of system parameters.
    In \cite{Krikidis2018}, the study examines the average Age of Information (AoI) in a wireless power transfer sensor network. Additionally, \cite{chen2021optimizing} investigates the interplay of throughput/delay and AoI in a two-user multiple access channel with a single energy harvesting source. In \cite{EH_DRL_AoI2019}, the study examines the average AoI for status updates from an EH transmitter with a finite-capacity battery. The research investigates optimal scheduling policies under known channel and EH statistics. In cases of unknown environments, the authors propose an adaptive reinforcement learning algorithm to learn system parameters and update policies in real-time.
	In \cite{leng2019minimizing}, the study focuses on a cognitive radio system with a secondary user as an EH sensor, deciding between spectrum sensing and status updating in each time slot. The sequential decision-making problem is framed as a Partially Observable MDP (POMDP) and solved using dynamic programming, with an exploration of the optimal policy's structural properties. 
	Another study \cite{xu2023optimal} tackles real-time IoT applications using EH sensors, aiming to minimize the Age of Correlated Information (AoCI) at the data fusion center. The approach involves formulating the dynamic status update as a POMDP and introducing a DRL algorithm to solve the problem.
	The study \cite{JayanthDist2023} focuses on optimizing wireless communication of stochastic process samples to minimize distortion at the destination while maintaining a specified AoI and cost of actions. It introduces a stationary randomized policy (SRP) solution and highlights challenges related to rapid source changes and channel states. Additionally, a constrained POMDP formulation for the problem has been defined.
	The article \cite{dong2020energy} optimizes IoT systems by minimizing AoI and distortion through effective policies, including save-and-transmit and fixed power transmission. Causal EH information is addressed with an MDP for optimal policy. The study reveals that the optimal transmit power is a bivalued function of the current age and distortion. The authors of \cite{gindullina2021age} study an EH monitoring node managing updates from diverse sources with different energy consumption and AoI values. The objective is to minimize average AoI through optimal actions (requesting an update from a source or staying idle) formulated as an MDP, with the optimal policy determined using the Value Iteration algorithm.
	In \cite{hatami2022demand}, the focus is on minimizing on-demand AoI in a multi-user IoT energy harvesting network, using an MDP formulation. The study proposes an iterative algorithm for optimal status updates, with a low-complexity alternative for scenarios with numerous sensors. In \cite{hatami2021aoi}, the problem is tackled without transmission constraints, employing a model-free Q-learning method within an MDP framework. \cite{holm2021freshness} introduces a pull-based communication model using the Age of Information at Query (QAoI) metric in an MDP, determining the optimal status updating policy for a monitoring scenario with periodic queries from a server to an EH sensor at an edge node. 
    The paper \cite{hu2023aoi} investigates online scheduling in wireless-powered communication networks for IoT devices. It focuses on minimizing the Expected Weighted Sum Age of Information (EWSAoI) by proposing a Max-Weight policy based on Lyapunov optimization theory. This policy schedules sensor nodes to transmit their data to a mobile edge server efficiently, considering wireless power transfer and channel fading effects.
    Additionally, The work \cite{DelfaniVage2023} examines and optimizes a real-time IoT network, considering energy harvesting, caching, and gossiping. It focuses on minimizing the average Version AoI in a destination gossiping network while managing energy constraints for the EH sensor and responding to network requests, utilizing the MDP framework. The work \cite{rafiee2023adaptive} deals with updating information efficiently for an EH IoT receiver that interacts with a variable-rate information source. It aims to minimize the average AoI by optimizing when the receiver turns on or off. The study uses the MDP framework to find optimal scheduling policies and introduces a state-adapted waiting policy.

	In this study, we demonstrate the close connection between the challenge of reserving energy for \emph{critical} use and the issue of ensuring timely status updates. Specifically, we examine an energy harvesting (EH) status update system that monitors a stochastic process with two states, a \emph{normal} state and an \emph{alarm} state. This framework encompasses systems where events occur with a certain probability at defined time intervals during normal operation, while the probability increases significantly during alarm operation. For instance, this scenario is applicable to networks, where the rate of packet arrivals during a denial of service attack contrasts with normal operation. Additionally, our focus is on systems where the demand for fresh status updates is considerably higher during alarm periods than in normal operation. To address this heightened demand, the system needs to account for the characteristics of the energy arrival process and strategically reserve energy when feasible.
	
	To the best of our knowledge this is the first work to consider an AoI-based status update system for a two-state stochastic process and study the impact of constrained energy resources on the optimal status update transmission policies.
	For an effective representation of the problem, we introduce two AoI variables, each corresponding to a state of the stochastic process. We expand the AoI definition to encompass scenarios where the state of the stochastic process changes without the monitoring application being informed of the change. Finally, our results illustrate how the optimal policy is influenced by the probabilities of energy harvesting, successful status update transmission, and the probability of the monitored process changing state from its current state.
	
	\section{System Model}
	\label{sec:systemModel}
	The system we consider is presented in Fig.~\ref{fig:systemDiagram} and comprises an Energy Harvesting sensor responsible for monitoring a stochastic process and sending status updates to a destination node, denoted as Rx. We assume that Rx is one hop away from the sensor, time is slotted, and each slot has a duration of T.
	\begin{figure}[h]
		\centering
		\includegraphics[width=4.9in]{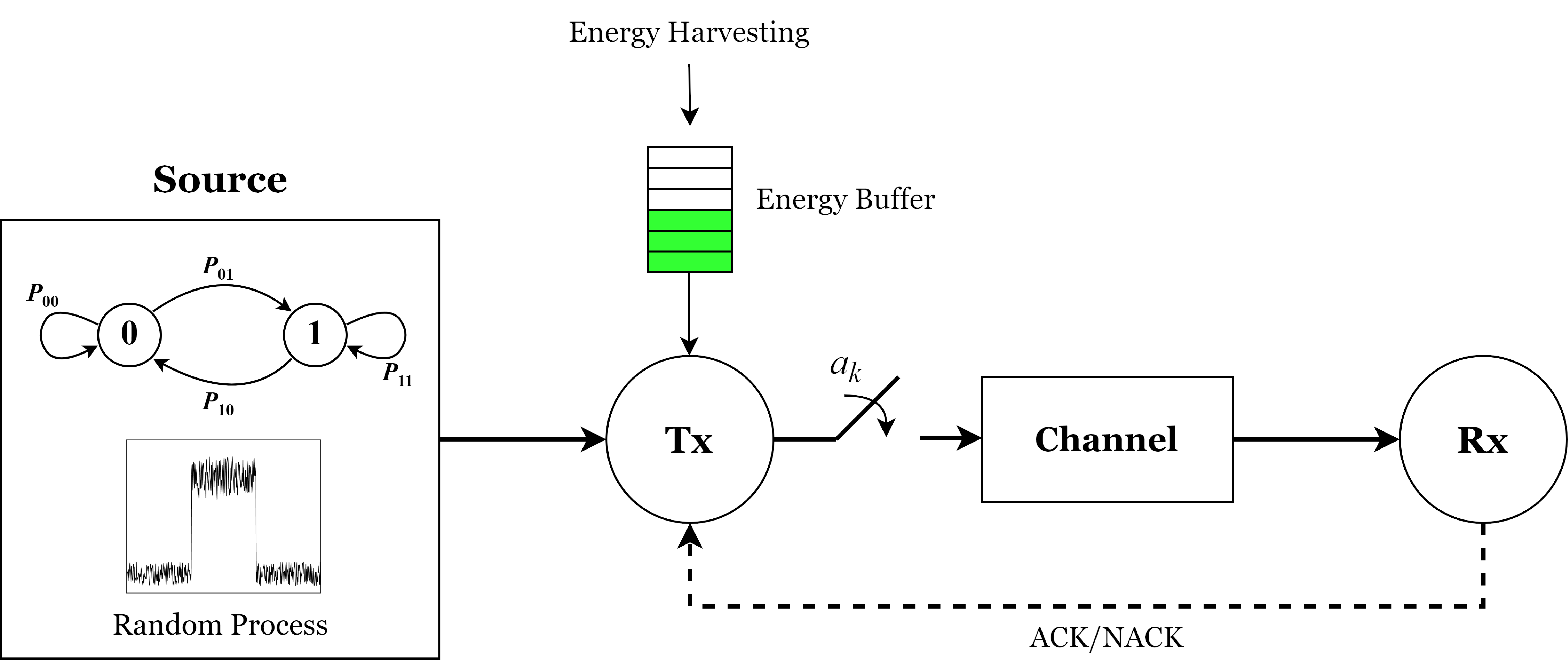}
		\caption{An EH status update system for a stochastic process with normal and alarm states.}
		\vspace{-20pt}
		\label{fig:systemDiagram}	
	\end{figure}
	At the beginning of each time slot, the stochastic process can exist in one of two states. The first state, $0$, indicates a \emph{normal} operational state. In contrast, the second state, $1$, signifies an \emph{alarm} operational state. An illustrative representation of such a process is presented in Fig.~\ref{fig:systemDiagram}. It is anticipated that a monitoring application for this stochastic process should deliver more frequent status updates during alarm periods.
	Let $\{Z_k\}$, where $k=0, 1, \dots$, represent the sequence of states of the stochastic process over time. We assume that the state of the stochastic process remains constant throughout a time slot. At the onset of the ($k+1$)-th time-slot, the state of the stochastic process transitions from $Z_k = z$ to $Z_{k+1} = z'$, governed by transition probabilities $P_{zz'}$, where $z, z' \in \{0,1\}$, as depicted in Fig.~\ref{fig:systemDiagram}.

	At the beginning of each time slot, the sensor generates a new status update and subsequently decides whether to transmit it to the destination. The sensor has an energy buffer capable of storing an integer number of energy units, with a maximum capacity of $E_{max}$ energy units. The sensor has a probability $P_e$ of harvesting an energy unit in a given time slot.
	We assume that each status update transmission consumes one energy unit, and no transmission is possible if the energy buffer is empty. For the purposes of this study, we do not consider energy costs associated with other sensor functions, such as sensing, processing, and data storage in memory. Each transmission has an independent probability of success, denoted as $P_s$, and this probability is unaffected by the outcomes of previous transmissions. Additionally, we assume that acknowledgment of a packet transmission occurs instantaneously.
	
	We employ the AoI metric to quantify the timeliness of status updates reaching the destination. AoI, as defined in \cite{kaul2012real}, represents the time elapsed since the generation of the last successfully decoded status update. However, our study must also account for state changes in the stochastic process. The destination remains unaware of any such state change until it receives a fresh status update. Additionally, the sensor node faces the challenge of deciding when to transmit a new status update, considering both the increased (or decreased) demand during alarm (normal) states of the stochastic process and the limited energy resources in the buffer. The sensor must leverage its knowledge of the stochastic process's state changes and the AoI value at the destination to achieve this objective.
	
	To address this scenario, we employ two distinct AoI variables, each corresponding to a different state of the stochastic process. We represent the AoI for the $z$-th state of the stochastic process at time $k$ as $\Delta_k^z, z\in\{0,1\}$.
	Additionally, we use the sequence of time indices where a state change occurs, denoted as $\{\tau_n: Z_{\tau_n} \neq Z_{\tau_n - 1},\ n=1, 2, ...\}$, and define $\tau_N$ as the time index of the most recent state change for the stochastic process by time $k$, where $N = \max \{n : \tau_n < k\}$.
	Lastly, let $Z_k^d$ represent the state that the destination \emph{knows} as the process's state at time $k$, indicating the state of the stochastic process included in the most recent status update received by the destination.
	The definition of AoI is then as follows,
	\begin{equation}
		\label{eq:AoiDefinition}
		\Delta_k^z = \begin{cases}
			\min\{k - U_k, \Delta_{max}^z\}, & \text{if $z = Z_k^d$}, \\
			\min\{k - \tau_N, \Delta_{max}^z\}, & \text{if $z \neq Z_k^d$ and $z = Z_k$}, \\
			0, & \text{if $z \neq Z_k^d$ and $z \neq Z_k$},
		\end{cases}
	\end{equation}
	where $U_k$ denotes the timestamp of the most recent packet received at the destination by time $k$, and $\Delta_{max}^z$ represents the maximum value of AoI associated with the highest level of staleness.
	
	The first branch of (\ref{eq:AoiDefinition}) applies to the AoI variable associated with the state of the stochastic process known at the destination by time $k$, aligning with the definition of AoI as presented in~\cite{kaul2012real}.
	The second branch of (\ref{eq:AoiDefinition}) is applicable in scenarios where \emph{one or more state changes} have occurred, leading to the current state of the stochastic process differing from the one recognized at the destination ($Z_k \neq Z_k^d$). In such instances, the AoI for $z = Z_k$, denoted as $\Delta_k^z$, is defined as the time elapsed since the last state change ($\tau_N$).
	Finally, the third branch of the equation applies for AoI $\Delta_k^z$ when $z$ is neither the state known by the destination nor the currently active state, i.e., the state known to the destination at the $k$-th time slot, $Z_k^d$, is equal to the actual state of the stochastic process. In such a case, the state $z$ that is not currently active, i.e., $z \neq Z_k^d$, is assigned an AoI value of zero.
	
	An illustration in Fig.~\ref{fig:AoiDefinition} depicts the evolution of $\Delta_k^0$ and $\Delta_k^1$ over time. Status updates occur at $t_k$ ($k \geq 0$), reaching the destination at time points denoted as $t_k'$. $\tau_c$ ($c \geq 0$) indicates times of stochastic process state changes.
	At $k=2$, the destination receives a status update indicating $Z(0) = 0$, causing $\Delta^0$ to increase while $\Delta^1$ remains zero. At $k=5\ (\tau_1)$, the process shifts to $Z(5)=1$, incrementing both $\Delta_k^0$ and $\Delta_k^1$ following the first and second branches of (\ref{eq:AoiDefinition}), respectively. At $k=9$, the destination receives an update confirming a state change at $\tau_1=5$, resetting $\Delta_k^0$ to zero according to the third branch of (\ref{eq:AoiDefinition}), and continuing the increment of $\Delta_k^1$ following the first branch. Finally, at $k=12\ (\tau_2)$, another state change occurs, repeating the process.
	
	\begin{figure}[htb!]
		\centering
		\begin{subfigure}[b]{\textwidth}
			\centering
			\includegraphics[scale=0.75]{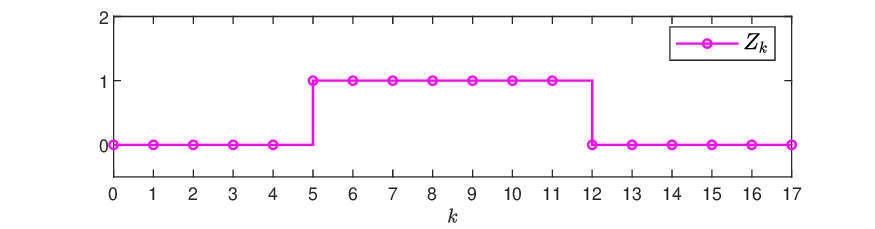}
			\caption{Time evolution of the stochastic process' state.}
			\label{fig:randomProcessStateEvolution}	
		\end{subfigure} 
		\begin{subfigure}[b]{\textwidth}
			\centering
			\includegraphics[scale=0.75]{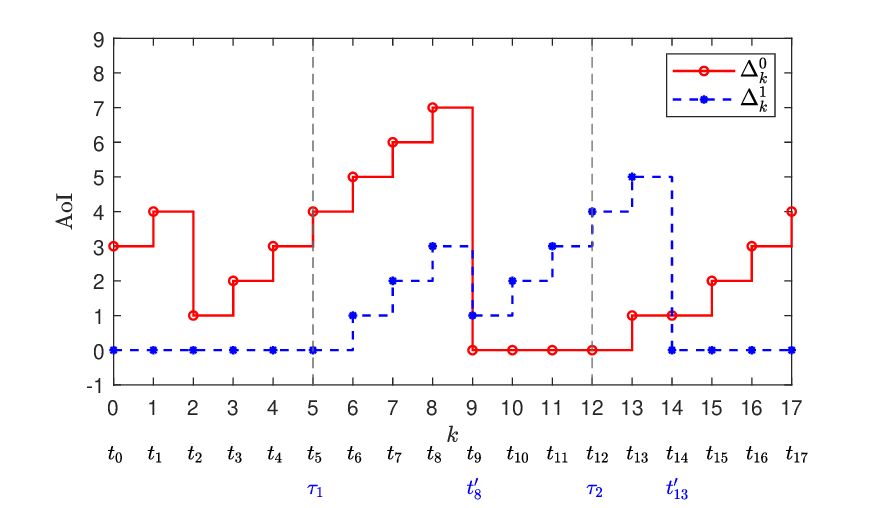}
			\caption{Time evolution of $\Delta_k^0$ and $\Delta_k^1$ for the $Z_k$ presented above.}
			\label{fig:AoI_and_cost_function}
		\end{subfigure}
		\caption{The first sub-figure presents the time evolution of the stochastic process' state. The second sub-figure presents the evolution of the AoI for each state of the stochastic process.}
		\vspace{-20pt}
		\label{fig:AoiDefinition}
	\end{figure}
	
	By employing these two AoI variables, we will be able to formulate various metrics or cost (reward) functions related to the staleness (freshness) of a system in two distinct states with different demands. The subsequent section will further elucidate these metrics.
	
	\section{Problem Formulation}
	\label{sec:problemFormulation}
	In this section, we present the state, action, and random variable spaces of the system, as well as the system's transition and cost functions.
	
	\subsubsection*{States}
	At the beginning of the $k$-th time-slot the state of the system is represented by the following state vector,
	\begin{equation}
		s_k = [Z_k, Z_k^d, E_k, \Delta_k^0, \Delta_k^1]^T,
	\end{equation}
	where $Z_k\in\{0, 1\}$ represents the state of the stochastic process, $Z_k^d$ signifies the state known by the destination at time $k$, $E_k = \{0, 1, ..., E_{max}\}$ is the energy in the buffer, and $\Delta_k^z, z\in\{0,1\}$ is the AoI at the destination for the $z$-th state of the stochastic process, with $T$ denoting the transpose operator. The set of all system states is denoted as $S$.
	
	\subsubsection*{Actions}
	When at least one energy unit is in the buffer, the sensing node can choose to transmit a fresh status update or conserve energy for later use. The action taken by the sensing node is denoted as $a_k \in {0,1}$, where $0$ indicates not transmitting a status update, and $1$ indicates transmitting one. If the energy buffer is empty, the sensor is restricted to action $0$. We use $a_s^*$ to denote the optimal action in state $s$, $A$ for the set of all actions, and $A(s)$ to represent the set of permissible actions at state $s$.
	\subsubsection*{Random variables}
	Given the system's state and the sensor's action, a stochastic transition to a new state occurs, determined by three random variables. The first, $W_k^s \in \{0,1\}$, signifies the random event of a successful transmission over the noisy channel, assumed to happen with probability $P_s$. If the sensor opts not to transmit at time-slot $k$, $W_k^s$ is forced to be zero. The second variable, $W_k^e \in \{0,1\}$, represents the random event of an energy unit arrival, assumed to occur with probability $P_e$ during a time slot. The third, $W_k^z \in \{0,1\}$, denotes the new state of the random process, determined by transition probabilities presented in Fig.\ref{fig:monitoredSystem}. These random variables' values become known to the sensor at the end of the $k$-th time-slot, as typical in optimal control theory~\cite{B12}. Lastly, we assume independence among $W_k^s$, $W_k^e$, and $W_k^z$, with their values being independent of previous time slots and identically distributed across all time slots. The random column vector $W_k = [W_k^s, W_k^e, W_k^z]^T$ collectively refers to the system's random variables.
	
	\subsubsection*{System Dynamics}
	\label{sec:systemDynamics}
	Given the current state of the system $s_k=\left[Z_k,Z_k^d,E_k,\Delta_k^0,\Delta_k^1\right]^T$ and the action $a_k$, the next state of system $s_{k+1}=\left[Z_{k+1},Z_{k+1}^d,E_{k+1},\Delta_{k+1}^0,\Delta_{k+1}^1\right]^T$ is determined by the realization of random vector $W_k=\left[W_k^s,W_k^e,W_k^z\right]$. 
	More specifically the state of the stochastic process at the $(k+1)$-th time-slot is provided by the random variable $W_k^z$ whose value becomes known by the end of the $k$-th time-slot. 
	\begin{align}
		Z_{k+1}=W_k^z,
	\end{align}
	while the state of the stochastic process known by the destination assumes a new value only in the case of a successful status update transmission,
	\begin{align}
		Z_{k+1}^d=
		\begin{cases}
			Z_{k}^d & W_k^s=0, \\
			Z_{k} & W_k^s=1.
		\end{cases}
	\end{align}
	
	The energy stored in the energy buffer at the beginning of the $(k+1)$-th time-slot depends on whether the sensor transmitted a status update and an energy unit was harvested during the $k$-th time-slot,
	\begin{align}
		E_{k+1}=E_{k}+W_k^e-a_k.
	\end{align}
	
	Here, we present a recursive definition for $\Delta_{k+1}$, although the evolution of the AoI variables over time was described in (\ref{eq:AoiDefinition}),
	\begin{align}
		\Delta_{k+1}^z=
		\begin{cases}
			0 & (z \neq Z_k,\ z \neq Z_k^d,\ W_k^s=0) \text{ or } (z \neq Z_k,\ W_k^s=1), \\
			\min\left\{\Delta_{k}^z+1,\Delta_{max}\right\} & (z=Z_k=Z_k^d,\ W_k^s=0) \text{ or } (Z_k \neq Z_k^d,\ z\in\{0,1\},\  W_k^s=0), \\
			1 & (z = Z_k,\ W_k^s=1).
		\end{cases}
	\end{align}
	
	\subsubsection*{Transition Probabilities}
	The transition probability of the system can be represented by the total probability theorem as follows:
	\begin{align}
		\label{TransProbMain}
		P\left[s_{k+1}|s_{k},a_k\right] &= \sum_{W_k}{P\left[s_{k+1},W_k|s_{k},a_k\right]} \\ \notag
		&= \sum_{W_k}{P\left[s_{k+1}|s_{k},a_k,W_k\right]P\left[W_k|s_{k},a_k\right]} \\ \notag
		&= \sum_{[W_k^s,W_k^e,W_k^z]}{P\left[s_{k+1}|s_{k},a_k,W_k^s,W_k^e,W_k^z\right]P\left[W_k^s,W_k^e,W_k^z|s_{k},a_k\right]}.
	\end{align}
	
	We can simplify the conditional probabilities in \ref{TransProbMain}:
	\begin{align}
		\label{ConditionalTrans}
		P\left[s_{k+1}|s_{k},a_k,W_k^s,W_k^e,W_k^z\right]&=P\left[Z_{k+1},Z_{k+1}^d,E_{k+1},\Delta_{k+1}^0,\Delta_{k+1}^1|s_{k},a_k,W_k^s,W_k^e,W_k^z\right] \\ \notag
		& = P\left[Z_{k+1}|W_k^z\right]\times P\left[Z_{k+1}^d|Z_{k}^d,W_k^s\right]\times P\left[E_{k+1}|E_{k},a_k,W_k^e\right] \\ \notag &\times P\left[\Delta_{k+1}^0|Z_{k},Z_{k}^d,\Delta_{k}^0,W_k^s\right]\times P\left[\Delta_{k+1}^1|Z_{k},Z_{k}^d,\Delta_{k}^1,W_k^s\right],
	\end{align}
	\begin{align}
		\label{ConditionalRvs}
		P\left[W_k^s,W_k^e,W_k^z|s_{k},a_k\right]&=P\left[W_k^s|s_{k},a_k\right]P\left[W_k^e|s_{k},a_k\right]P\left[W_k^z|s_{k},a_k\right]\\ \notag
		&=P\left[W_k^s|E_k,a_k\right]P\left[W_k^e\right]P\left[W_k^z|Z_k\right],
	\end{align}
	where:
	
	{\small
		\begin{align}
			\label{EvolutionProbZ}
			P\left[Z_{k+1}|W_k^z\right]=
			\begin{cases}
				1 & Z_{k+1}=W_k^z,\\
				0 & \text{otherwise,}
			\end{cases}
		\end{align}
		\begin{align}
			\label{EvolutionProbZd}
			P\left[Z_{k+1}^d|Z_{k}^d,W_k^s\right]=
			\begin{cases}
				1 & Z_{k+1}^d=Z_{k}^d, W_k^s=0,\\
				1 & Z_{k+1}^d=Z_{k}, W_k^s=1,\\
				0 & \text{otherwise,}
			\end{cases}
		\end{align}
		\begin{align}
			\label{EvolutionProbE}
			P\left[E_{k+1}|E_{k},a_k,W_k^e\right]=
			\begin{cases}
				1 & E_{k+1} = E_{k}+W_k^e-a_k,\\
				0 & \text{otherwise,}
			\end{cases}
		\end{align}
		
		\begin{align}
			\label{EvolutionProbDz}
			P&\left[\Delta_{k+1}^z|Z_{k},Z_{k}^d,\Delta_{k}^z,W_k^s\right] \\ \notag
			&=
			\begin{cases}
				1 & \Delta_{k+1}^z=0,\ (z \neq Z_k,\ z \neq Z_k^d,\ W_k^s=0) \text{ or } (z \neq Z_k,\ W_k^s=1), \\
				1 & \Delta_{k+1}^z=\min\left\{\Delta_{k}^z+1,\Delta_{max}\right\},\ (z=Z_k=Z_k^d,\ W_k^s=0) \text{ or } (Z_k \neq Z_k^d,\ z\in\{0,1\},\  W_k^s=0), \\
				1 & \Delta_{k+1}^z=1,\ (z = Z_k,\ W_k^s=1).
			\end{cases}
		\end{align}
		\begin{align}
			\label{RvProbWs}
			P\left[W_k^s|E_k,a_k\right]=
			\begin{cases}
				1 & W_k^s=0,\ \left(a_k=0 \text{ or } E_k=0\right), \\
				P_s & W_k^s=1,\ a_k=1, \\
				1-P_s & W_k^s=0,\ a_k=1,
			\end{cases}
		\end{align}
		\begin{align}
			\label{RvProbWe}
			P\left[W_k^e\right]=
			\begin{cases}
				P_e & W_k^e=1, \\
				1-P_e & W_k^e=0,
			\end{cases}
		\end{align}
		\begin{align}
			\label{RvProbWz}
			P\left[W_k^z|Z_k\right]=
			\begin{cases}
				P_{00} & W_k^z=0,\ Z_k=0, \\
				P_{01} & W_k^z=1,\ Z_k=0, \\
				P_{10} & W_k^z=0,\ Z_k=1, \\
				P_{11} & W_k^z=1,\ Z_k=1.
			\end{cases}
		\end{align}
	}
	
	By substituting equations \eqref{EvolutionProbZ} to \eqref{EvolutionProbDz} into \eqref{ConditionalTrans}, and equations \eqref{RvProbWs} to \eqref{RvProbWz} into \eqref{ConditionalRvs}, the transition probability \eqref{TransProbMain} is determined. 
	
	\subsubsection*{Transition cost function}
	We define a general metric as the cost function as follows:
	\begin{align}
		\label{eq:metric}
		g(s_k, a_k, w_k) = \big(1-Z_k\big)\cdot f(\Delta_k^0) + Z_k \cdot h\left( \Delta_k^1\right),
	\end{align}
	
	where $f(\cdot)$ and $h(\cdot)$ are two real-valued functions defined on non-negative integers, with the condition that $h(\cdot)$ ages faster than $f(\cdot)$, i.e., $h(\Delta_k) \geq f(\Delta_k), \forall \Delta_k \in \{0,1,2,\cdots\}$.
	Here, for simplicity, we consider the linear and square functions for $f(\cdot)$ and $h(\cdot)$, respectively, where the cost associated with each state transition is given by,
	\begin{align}
		\label{eq:transitionCost}
		g(s_k, a_k, w_k) = g(s_k, a_k) = \big(1-Z_k\big)\cdot \Delta_k^0 + Z_k \cdot \big( \Delta_k^1\big)^2,
	\end{align}
	where $w_k$ is the realization of random vector $W_k$ at the $k$-th time-slot.
	From (\ref{eq:transitionCost}) we observe that when the stochastic process is in the normal state ($Z_k = 0$), the transition cost increases linearly with $\Delta_k^0$, while when the system is in the alarm state ($Z_k = 1$) the transition cost increases with the square of $\Delta_k^1$.
	Thus, the transition cost function captures the increased demand for status updates when $Z_k = 1$.
	
	\subsubsection*{Total cost function}
	We aim to minimize the cumulative cost over an infinite time span,
	\begin{equation} 
		\label{eq:cummulativeCostFunction}
		J_{\mu}(s_0) = \underset{N \to \infty}{\lim} \underset{\underset{k=0,1,\dots}{W_k,}}{\mathop{\mathbb{E}}} \left\lbrace \sum_{k=0}^{N-1} \gamma^k g(s_k, a_k, w_k) | s_0\right\rbrace,  
	\end{equation}
	where $s_0$ denotes the initial state of the system, the expectation $\mathbb{E}\lbrace \cdot \rbrace$ is computed based on the joint probability distribution of random variables $W_k$ for $k\in \{0, 1, \dots\}$, and $\gamma$ serves as a discount factor (where $0 < \gamma < 1$), indicating diminishing importance of induced cost over time. Lastly, let $\mu=\{u_0,u_1,u_2,\cdots,u_k,\cdots\}$ represent a deterministic policy mapping each state $x_k$ to a specific action $a_k=u_k (x_k)$ at each time slot $k$. 
	
	\emph{Our objective} is to obtain an optimal policy $\mu^*=\left\{u^*_0,u^*_1,u^*_2,\cdots\right\}$ that minimizes (\ref{eq:cummulativeCostFunction}).

	\section{Analytical Results}
	
	\subsection{Optimal Policy}
	\label{sec:optimalPolicy}
	
	The dynamic system outlined in section~\ref{sec:problemFormulation} is characterized by finite state, control, and probability spaces. State transitions rely on $s_k$, $a_k$, and $w_k$, independent of their previous values. Furthermore, the probability distribution of random variables remains constant over time. The cost linked to a state transition is bounded, and the cost function $J(\cdot)$ accumulates additively over time. These structural characteristics establish the considered dynamic system as a Markov Decision Process (MDP), where the state transition probabilities completely describe its dynamics. Specifically, the problem \eqref{eq:cummulativeCostFunction} is an infinite horizon discounted cost MDP problem with bounded cost per stage~\cite[Sec. 1.2]{B12}.
	For the MDP under consideration, given that $0<\gamma<1$, there exists an optimal stationary policy $\mu^*$
	which is characterized by Bellman's equation~\cite[Prop. 1.2.5, pg. 17]{B12}. Specifically, when the system is in state $s$, the optimal stationary policy $\mu^*$ always applies the same control $a^\ast(s)$ that minimizes (\ref{eq:cummulativeCostFunction}), i.e., 
	\begin{gather}
		\mu^* = \arg \underset{\mu \in \mathcal{M}}{\min} J_{\mu}(s), 
	\end{gather}
	where $a^\ast(s) = u^*(s) ,\ \text{for all $s \in S$}$, and $\mathcal{M}$ is the set of all policies.
	Let $V^\ast(s)=J^*(s)$ be the infinite horizon discounted cost attained when the optimal policy $\mu^*$ is applied and the system begins at state $s$. The optimal cost $V^\ast(s)$ and the optimal action $a^\ast(s)$ satisfy the Bellman's equation:
	\begin{align}
		\label{BellmanEqnCost}
		V^\ast(s) &= \min_{a\in\{0,1\}}\left\{\sum_{\tilde{s}\in S}{P(\tilde{s}|s,a)\left[g(s,a)+\gamma V^\ast(\tilde{s})\right]}\right\}, \forall s \in S,
		\\ \label{BellmanEqnAction}
		a^\ast(s) &= \argmin_{a\in\{0,1\}}\left\{\sum_{\tilde{s}\in S}{P(\tilde{s}|s,a)\left[g(s,a)+\gamma V^\ast(\tilde{s})\right]}\right\}, \forall s \in S,
	\end{align}
	where $s=\left[Z,Z^d,E,\Delta^0,\Delta^1\right]^T$, $\tilde{s}=\left[\tilde{Z}_k,\tilde{Z}_k^d,\tilde{E}_k,\tilde{\Delta}_k^0,\tilde{\Delta}_k^1\right]^T$ and $V(s)$ is the value function of the MDP problem. 
	
	Given that the transition cost $g(s, a)$ is bounded and that $0< \gamma < 1$, the operator, 
	\begin{equation}
		\label{eq:VI}
		(TV)(s) = \underset{a\in\{0,1\}}{\min} \left\{\sum_{\tilde{s}\in S} P(\tilde{s}|s,a)\left[g(s,a)+\gamma V(\tilde{s})\right] \right\},
	\end{equation} 
	is a contraction mapping~\cite[Assumption D, Prop. 1.2.1, pg. 14]{B12} and starting with an arbitrarily initialized vector $V(s), s \in S$, and repeatedly applying transformation $(TV)$ for all states $s \in S$ we attain the optimal cost $V^\ast$ and at the same time derive the optimal policy $\mu^\ast$ for all $s\in S$ according to~\cite[Prop. 1.2.1, pg. 14]{B12} which states that,
	\begin{equation}
		V^*(s) = \lim_{m\rightarrow\infty} (T^mV)(s), 
	\end{equation} 
	where $(T^mV)(s) = (T(T^{m-1}\dots(T^0 V))(s)$ and $(T^0V)(s) = V(s)$.
	(\ref{eq:VI}) is a formal description of the Value Iteration (VI) algorithm~\cite[Section 2.2, pg. 84]{B12}.

	\subsection{Threshold Policy}
    \begin{definition}
	   Policy $\mu$ is a threshold policy if for each combination of values for $Z$, $Z^d$, and $E$ there exists a threshold $\delta_T=(\Delta_T^0,\Delta_T^1)$ such that then sensor will transmit, i.e. $a(s)=1$, only if $\delta=(\Delta^0,\Delta^1) \geq (\Delta_T^0,\Delta_T^1)=\delta_T$, where $\geq$ is meant to hold element-wise.
    \end{definition}

    \begin{theorem}
	   An optimal policy of the MDP problem is a threshold policy.
    \end{theorem}
    \begin{proof}
    	The proof can be found in appendix \ref{ProofTheorem1}.
    \end{proof}

	\section{Numerical Results}
	In this section, we conduct a numerical evaluation of the optimal infinite horizon discounted cost, $J^*(\cdot)$, under different system parameter configurations. For consistency across all experiments, we fix the discount factor at $\gamma = 0.99$, set both AoI variables' upper bounds ($\Delta_{max}^0$ and $\Delta_{max}^1$) to 10, and for ease of interpretation, assume a constant initial state $s_0$ for the system. More specifically, we assume the deployment of the sensor during the normal state $(Z_k = 0)$ of the random process, with this information known to the destination $(Z_k^d = 0)$. The energy buffer starts empty $(E_k = 0)$, and the initial state $s_0$ is defined as $[0, 0, 0, 1, 0]^T$ with AoI counters $\Delta_k^0$ and $\Delta_k^1$ set to $1$ and $0$, respectively.
	
	\begin{figure}[!htb]
		\centering
		\includegraphics[scale=0.6]{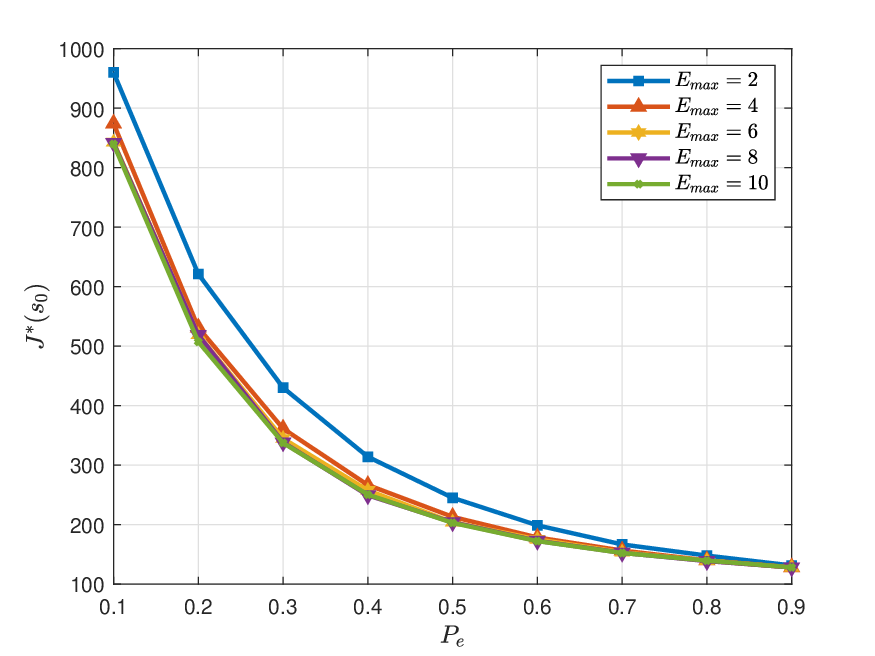}
		\caption{Impact of energy buffer's capacity, $E_{max}$, on $J^*(s_0)$.}
		\label{fig:ehProbVsCostVsBufferSize}
	\end{figure}
	In Fig.~\ref{fig:ehProbVsCostVsBufferSize}, we display $J^*(s_0)$ across various capacities of the energy buffer $E_{max}$ while the energy harvesting probability $P_e$ varies. In all experiments depicted in Fig.~\ref{fig:ehProbVsCostVsBufferSize}, the state transition probabilities of the stochastic process were configured as follows:
	\begin{equation}
		P_z = \begin{bmatrix}
			\label{eq:transitionProbabilityMatrix}
			0.9 & 0.1 \\
			0.2 & 0.8 
		\end{bmatrix},
	\end{equation}
	and the transmission success probability $P_s$ was set to $0.8$. 
	Fig.~\ref{fig:ehProbVsCostVsBufferSize} illustrates that being in an environment with a high probability of energy harvesting and having a larger capacity energy buffer contributes positively to reducing $J^*(s_0)$. The results in Fig.~\ref{fig:ehProbVsCostVsBufferSize} also indicate that the influence of the energy buffer's capacity on $J^*(s_0)$ becomes negligible when $E_{max}$ exceeds a certain threshold for the given system configuration.
	
	\begin{figure}[!htb]
		\centering
		\includegraphics[scale=0.6]{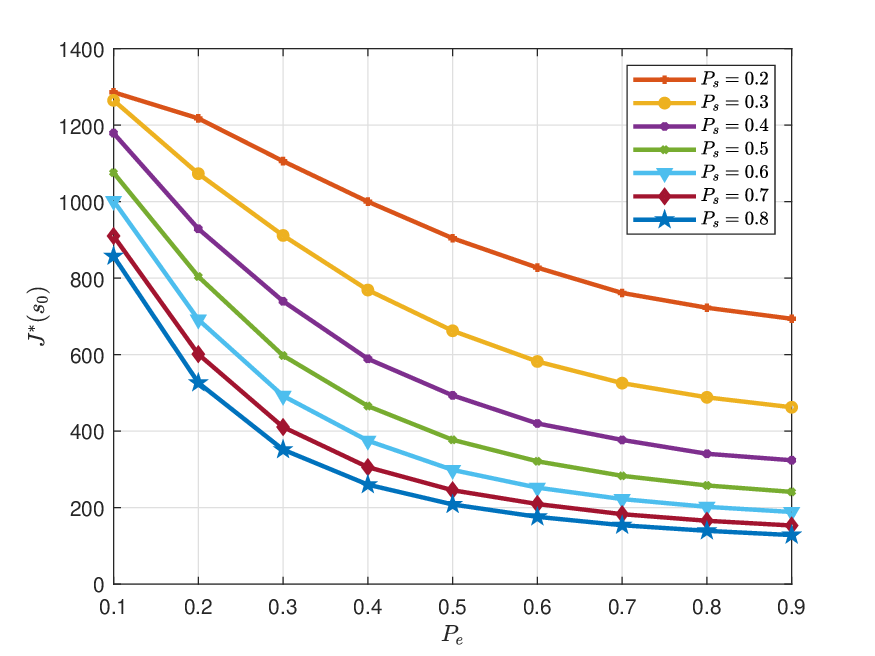}
		\caption{Impact of transmission success probability $P_s$ on $J^*(s_0)$.}
		\label{fig:ehProbVsCostVsTxSuccessProb}
	\end{figure}
	Fig.~\ref{fig:ehProbVsCostVsTxSuccessProb} illustrates $J^*(s_0)$ in relation to $P_e$ for various transmission success probabilities $P_s$. In this series of experiments, $P_z$ corresponds to the matrix defined in (\ref{eq:transitionProbabilityMatrix}), and $E_{max}$ was fixed at $5$. The figure indicates that an increase in $P_s$ consistently leads to a reduction in $J^*(s_0)$. Additionally, the results suggest that, given the energy buffer's capacity, targeting a higher $P_s$ value in environments with a low probability of energy harvesting is advisable.

	Fig.~\ref{fig:transitionProbabilitiesVsCost} depicts $J^*(s_0)$ for various combinations of state transition probabilities governing the stochastic process. In the figure, the probability $P_{01}$ (or $P_{10}$) represents the probability of the stochastic process transitioning from the normal (alarm) state to the alarm (normal) state by the end of a time slot. The probabilities $P_{00}$ and $P_{11}$ are calculated as $1 - P_{01}$ and $1 - P_{10}$ respectively.
	\begin{figure}[!htb]
		\centering
		\includegraphics[scale=0.6]{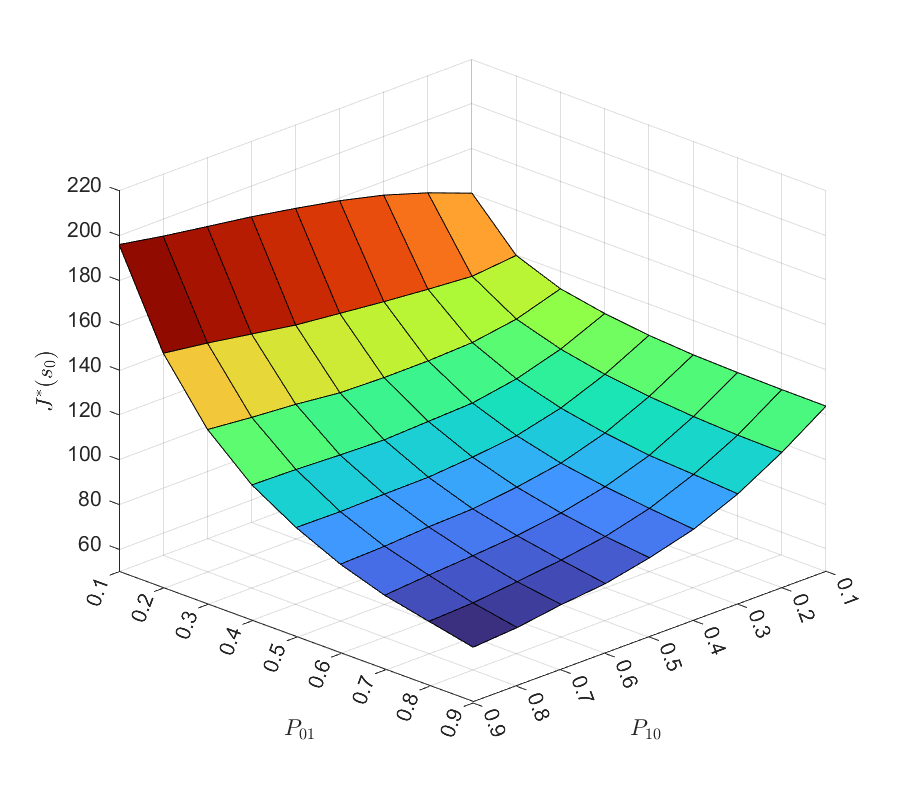}
		\caption{Impact of different stochastic process's state transition probabilities on $J^*(s_0)$.}
		\vspace{-20pt}
		\label{fig:transitionProbabilitiesVsCost}
	\end{figure}
	The highest value of $J^*(s_0)$ is observed when $(P_{01}, P_{10})=(0.9, 0.1)$, indicating a situation where the stochastic process is highly likely to transition from the normal to the alarm state and, once in the alarm state, has a low probability of returning to the normal state. The mentioned cost decreases as the probability of returning to the normal state, $P_{10}$, increases.
	One might anticipate a similar cost reduction when decreasing the values of $P_{01}$; however, the results in Fig.~\ref{fig:transitionProbabilitiesVsCost} demonstrate that decreasing $P_{01}$ could lead to an increase in $J^*(s_0)$. To be specific, the minimum value of $J^*(s_0)$ is observed when $(P_{01}, P_{10})=(0.9, 0.9)$, and $J^*(s_0)$ actually rises as $P_{01}$ decreases from $0.9$ to $0.1$. Initially, this may seem counterintuitive, as one might expect that when $P_{01}$ is small, the system will spend less time in the alarm state, resulting in a smaller cost $J^*(s_0)$. However, the underlying logic behind this phenomenon is that when $P_{01}$ and $P_{10}$ are large, the stochastic process spends only a limited number of time slots in each state.
	If the transmission success probability $P_s$ is high, neither $\Delta_k^0$ nor $\Delta_k^1$ will reach significant values, resulting in low transition costs as defined by (\ref{eq:transitionCost}).
	
	\begin{figure}[htb!]
		\centering
		\begin{subfigure}[]{0.49\columnwidth}
			\centering
			\includegraphics[width=\linewidth,trim={0.8cm 0cm 0.8cm 0.0cm}]{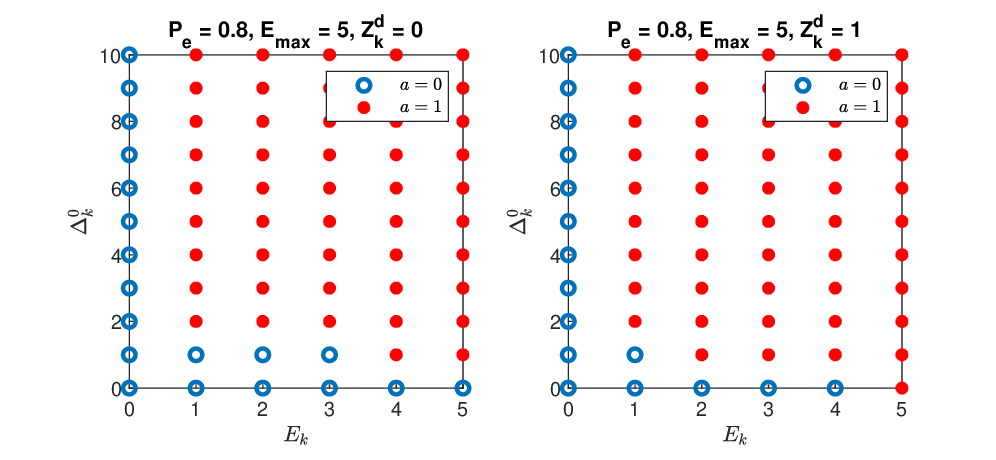}
			\vspace{-30pt}
			\caption{$Z_k = 0, P_e = 0.8, \Delta_k^1 = 0$}
			\vspace{6pt}
			\label{fig:Pe08z0}	
		\end{subfigure} 
		\begin{subfigure}[]{0.49\columnwidth}
			\centering
			\includegraphics[width=\linewidth,trim={0.8cm 0cm 0.8cm 0.0cm}]{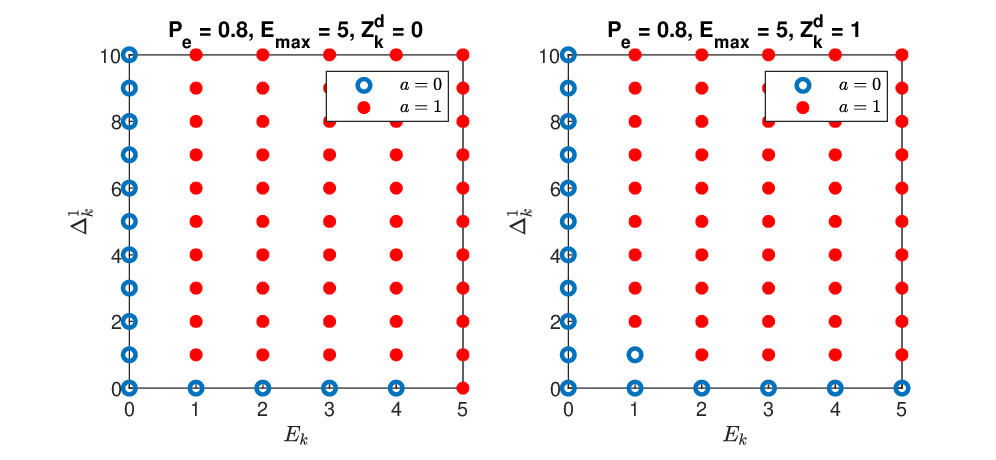}
			\vspace{-30pt}
			\caption{$Z_k = 1, P_e = 0.8, \Delta_k^0 = 0$}
			\vspace{6pt}
			\label{fig:Pe08z1}
		\end{subfigure} 
		\begin{subfigure}[]{0.49\columnwidth}
			\centering
			\includegraphics[width=\linewidth,trim={0.8cm 0cm 0.8cm 0.0cm}]{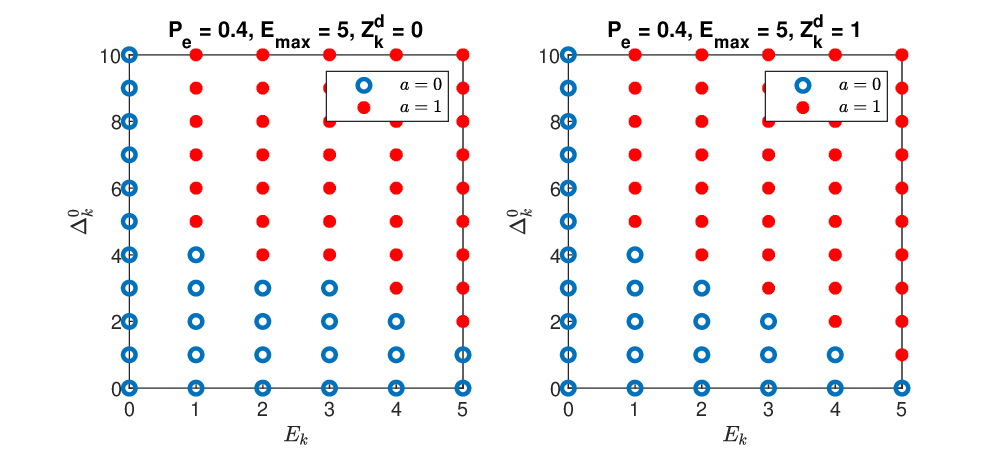}
			\vspace{-30pt}
			\caption{$Z_k = 0, P_e = 0.4, \Delta_k^1 = 0$}
			\vspace{6pt}
			\label{fig:Pe04z0}
		\end{subfigure}
		\begin{subfigure}[]{0.49\columnwidth}
			\centering
			\includegraphics[width=\linewidth,trim={0.8cm 0cm 0.8cm 0.0cm}]{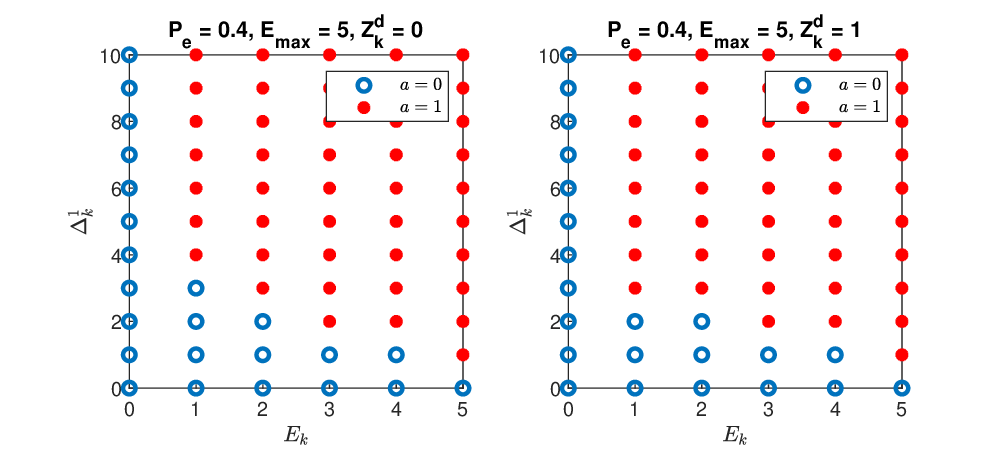}
			\vspace{-30pt}
			\caption{$Z_k = 1, P_e = 0.4, \Delta_k^0 = 0$}
			\vspace{6pt}
			\label{fig:Pe04z1}
		\end{subfigure}
		\caption{The optimal actions when the stochastic process is in the normal state ($Z_k=0$) or the alarm state $(Z_k=1)$ and $P_e$ is either $0.8$ or $0.4$.}
	\vspace{-20pt}
	\label{fig:policyActions}
\end{figure}
In Fig.~\ref{fig:policyActions}, we illustrate the optimal policy $\mu^*$ for two scenarios and two states of the stochastic process. In the first scenario, energy is harvested with a high probability at each time slot ($P_e = 0.8$). In contrast, in the second scenario, $P_e$ is set to a lower value of $0.4$. In both experiments, the transmission success probability $P_s$ was fixed at $0.8$, the energy buffer capacity was set to $5$, and the stochastic process' state transition probabilities $P_z$ were defined as in (\ref{eq:transitionProbabilityMatrix}). Specifically, Fig.~\ref{fig:Pe08z0} presents the optimal transmitter actions based on the number of energy units $E_k$ stored in the energy buffer and the value of the AoI counter $\Delta_k^0$ when $Z_k = 0$ and $P_e=0.8$. Figure~\ref{fig:Pe08z1} shows the corresponding results for the case where the stochastic process is in the alarm state ($Z_k = 1$). Figures~\ref{fig:Pe04z0} and~\ref{fig:Pe04z1} present the corresponding results for the second scenario with $P_e = 0.4$.

Comparing Figures~\ref{fig:Pe08z0} and~\ref{fig:Pe08z1}, it is evident that when the probability of harvesting energy is high, the actions prescribed by the optimal policy exhibit minimal differences between the two states of the stochastic process. Specifically, the optimal policy $\mu^*$ still tends to reserve energy when $Z_k=0$ by refraining from transmitting a status update ($a^*=0$) when $(E_k, \Delta_k^0) \in \{(2,1), (3,1)\}$, representing the only distinction between the two cases. \textit{The emphasis on energy reservation, anticipating alarm periods, becomes more pronounced when the probability of harvesting energy is lower}.
Comparing Figures~\ref{fig:Pe04z0} and~\ref{fig:Pe04z1}, we observe that the optimal policy restrains the transmitter from sending status updates when $Z_k =0$, even with a substantial number of energy units stored in the energy buffer. This strategy aims to avoid the quadratic cost associated with $Z_k^1$. The transition probability values of the stochastic process further support the justification for this optimal policy.
Matrix $P_z$ indicates that once the stochastic process enters an alarm state, it will likely remain in that state ($P_{11} = 0.8$). \textit{Therefore, reserving energy becomes essential to accommodate potentially extended periods during which the stochastic process remains in the alarm state}. 

\begin{figure}[htb!]
	\centering
	\begin{subfigure}[]{0.49\columnwidth}
		\centering
		\includegraphics[width=\linewidth,trim={0.8cm 0cm 0.8cm 0.0cm}]{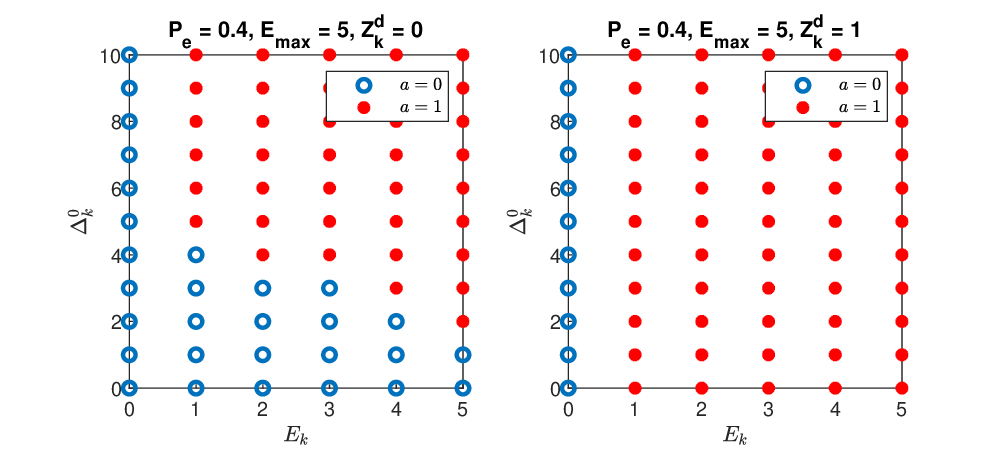}
		\vspace{-30pt}
		\caption{$Z_k = 0, P_e = 0.4, \Delta_k^1 = 10$}
		\vspace{6pt}
		\label{fig:Pe04z0_D10}
	\end{subfigure}
	\begin{subfigure}[]{0.49\columnwidth}
		\centering
		\includegraphics[width=\linewidth,trim={0.8cm 0cm 0.8cm 0.0cm}]{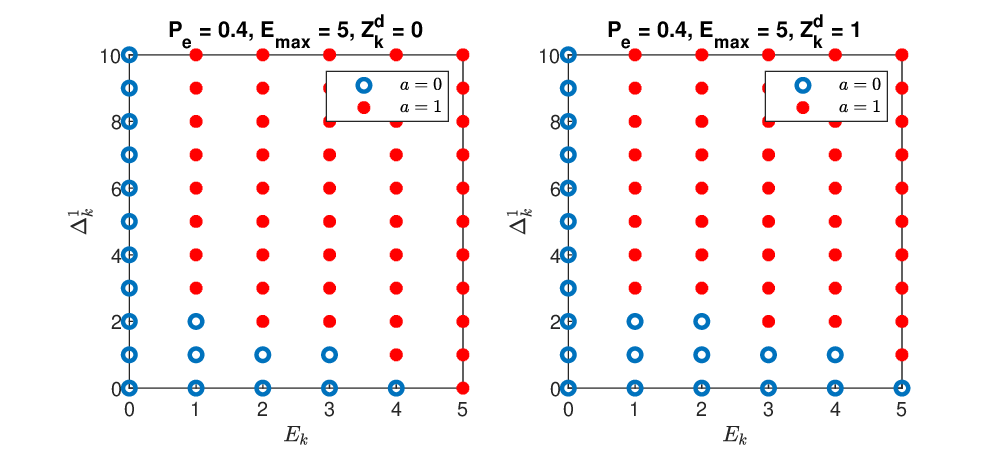}
		\vspace{-30pt}
		\caption{$Z_k = 1, P_e = 0.4, \Delta_k^0 = 10$}
		\vspace{6pt}
		\label{fig:Pe04z1_D10}
	\end{subfigure}
	\caption{The optimal actions when the stochastic process is in the normal state ($Z_k=0$) or the alarm state $(Z_k=1)$, $P_e=0.4$, and $\Delta_k^z=10$ for $z \neq Z_k$.}
	\vspace{-20pt}
	\label{fig:policyActions_D10}
\end{figure}

In Fig. \ref{fig:policyActions_D10}, we have also examined the situation in which the AoI variable related to a state other than the current state of the stochastic process, i.e., $\Delta_k^z$ for $z \neq Z_k$, has a high value. Comparing Figs. \ref{fig:Pe04z0_D10} and \ref{fig:Pe04z1_D10} with Figs. \ref{fig:Pe04z0} and \ref{fig:Pe04z1}, we find that when $Z_k$ and $Z_k^d$ are identical, this AoI variable becomes irrelevant. However, in cases where $Z_k \neq Z_k^d$, it influences optimal actions and reduces the AoI thresholds at various energy levels. This is because both AoI variables increase concurrently, making it reasonable to transmit fresh updates when one of them has a high level, particularly when $\Delta_k^1$ is elevated (as shown in the right figure in Fig. \ref{fig:Pe04z0_D10}).

\section{Conclusions and Future work}
This study examined a status update system incorporating an energy harvesting sensor to monitor a stochastic process. This process can exist in either a normal or an alarm operational state, each demanding different levels of timeliness. To address this challenge, we introduced a state-aware freshness metric characterized by a linear increase of age during the normal state and a quadratic increase during the alarm state. We then approached the optimization of this metric by formulating it as an MDP problem. The analytical demonstration revealed that the optimal policy structure is threshold-based. We then developed optimal policies for transmitting status updates across various system configurations. Through numerical assessments, we evaluated the influence of the energy buffer's capacity, transmission success probability, and the stochastic process' transition probabilities on the system's overall performance. 
Our next step includes using Deep Reinforcement Learning to tackle situations where the system's model is unknown, and $\Delta_{\text{max}}^z$ values are significantly large, resulting in a state space too large for tabular representation. Furthermore, this work can be extended to goal-oriented semantics-aware communication models, where accounting for receiver-side data utilization, actuation costs, or timeliness of actuation becomes essential.

\ifCLASSOPTIONcaptionsoff
\newpage
\fi

\bibliographystyle{IEEEtran}
\bibliography{bibliography}

\appendices

\section{Proof of Theorem 1}
\label{ProofTheorem1}
\begin{proof}
Since $g(s,a)=g(s)=(1-Z)\Delta^0+Z(\Delta^1)^2$ is not a function of $a$, the Bellman's equation can be simplified as follows:

\begin{align}
	V(s) &= g(s) + \argmin_{a\in\{0,1\}}\left\{\sum_{\tilde{s}\in S}{\gamma P(\tilde{s}|s,a)V(\tilde{s})}\right\}, \\
	a^\ast(s) &= \argmin_{a\in\{0,1\}}\left\{\sum_{\tilde{s}\in S}{P(\tilde{s}|s,a)V(\tilde{s})}\right\}.
\end{align}

We have dropped the asterisk superscript above $V$ for the sake of simplicity. Let us define $V^1(s)=\sum_{\tilde{s}\in S}{P(\tilde{s}|s,a=1)V(\tilde{s})}$, $V^0(s)=\sum_{\tilde{s}\in S}{P(\tilde{s}|s,a=0)V(\tilde{s})}$, and $\Delta V(s)=V^1(s)-V^0(s)$. Thus, we have:
\begin{align}
	a^\ast(s) = 
	\begin{cases}
		0 & \Delta V(s) \geq 0, \\
		1 & \Delta V(s) < 0.
	\end{cases}
\end{align}

In what follows, we show that $\Delta V(s)$ is a decreasing function of $(\Delta^0,\Delta^1)$ for each combination of $(Z,Z^d,E)$.  Thus, $\Delta V(s)$
can become negative for sufficiently large values of $(\Delta^0,\Delta^1)$, leading to the action $a = 1$ for $(\Delta^0,\Delta^1) \geq (\Delta_T^0,\Delta_T^1)$.
\begin{align}
	\label{DV_MainExp}
	\Delta V(s)=V^1(s)-V^0(s)=\sum_{\tilde{s}\in S}{\left[P(\tilde{s}|s,a=1)-P(\tilde{s}|s,a=0)\right]V(\tilde{s})}.
\end{align}

When $E=0$, then $\Delta V(s)=0$ for each combination of $(Z,Z^d,\Delta^0,\Delta^1)$, so the action $a=0$ is optimal. We therefore consider the cases where $E>0$. 
The second term in \eqref{DV_MainExp} can be determined using the equations \eqref{TransProbMain}, \eqref{ConditionalTrans}, and \eqref{ConditionalRvs}.
\begin{align}
	P(\tilde{s}|s,a=0)=\sum_{[W^s,W^e,W^z]}{P\left[\tilde{s}|s,a=0,W^s,W^e,W^z\right]P\left[W^s|E,a=0\right]P\left[W^e\right]P\left[W^z|Z\right]}.
\end{align}
If the sensor decides against a transmission, then $W^s=0$ with probability one, so we have:
\begin{align}
	\label{TransProbConda0}
	P(\tilde{s}|s,a=0)=\sum_{[W^e,W^z]}{P\left[\tilde{s}|s,a=0,W^s=0,W^e,W^z\right]P\left[W^e\right]P\left[W^z|Z\right]}.
\end{align}
According to \eqref{EvolutionProbZ} - \eqref{EvolutionProbDz}, $P\left[\tilde{s}|s,a=0,W^s=0,W^e,W^z\right]=1$ for those $\tilde{s}\in S$ that hold all of the following conditions; it is $0$ otherwise.

\begin{subequations}
	\begin{gather}
		\tilde{Z}=W^z, \\ 
		\tilde{Z}^d=Z^d, \\ 
		\tilde{E}=E+W^e, \label{NextEa0}\\ 
		\left(\tilde{\Delta}^0=0,Z \neq 0,Z^d \neq 0 \right) \text{or } \left(\tilde{\Delta}^0=\Delta^0+1,\ (Z=Z^d=0 \text{ or } Z \neq Z^d)\right)\!, \\
		\left(\tilde{\Delta}^1=0,Z \neq 1,Z^d \neq 1 \right) \text{or } \left(\tilde{\Delta}^1=\Delta^1+1,\ (Z=Z^d=1 \text{ or } Z \neq Z^d)\right)\!.
	\end{gather}
\end{subequations}
We omitted the $\min\{\cdot,\Delta_{max}\}$ term due to space constraints, as it does not affect the proof of the theorem. The first term in \eqref{DV_MainExp} can also be simplified using the equations \eqref{TransProbMain}, \eqref{ConditionalTrans}, and \eqref{ConditionalRvs}.
\begin{subequations}
	\label{a0NonZeroConds}
	\begin{align}
		\notag P&(\tilde{s}|s,a=1)=\sum_{[W^s,W^e,W^z]}{P\left[\tilde{s}|s,a=1,W^s,W^e,W^z\right]P\left[W^s|E,a=1\right]P\left[W^e\right]P\left[W^z|Z\right]} \\ \notag
		&=\sum_{[W^e,W^z]}{P\left[\tilde{s}|s,a=1,W^s=0,W^e,W^z\right]P\left[W^s=0|E,a=1\right]P\left[W^e\right]P\left[W^z|Z\right]} \\ \notag
		&+\sum_{[W^e,W^z]}{P\left[\tilde{s}|s,a=1,W^s=1,W^e,W^z\right]P\left[W^s=1|E,a=1\right]P\left[W^e\right]P\left[W^z|Z\right]} \\ 
		&=(1-P_s)\sum_{[W^e,W^z]}{P\left[\tilde{s}|s,a=1,W^s=0,W^e,W^z\right]P\left[W^e\right]P\left[W^z|Z\right]} \label{TransProbConda1Ws0} \\
		&+P_s\sum_{[W^e,W^z]}{P\left[\tilde{s}|s,a=1,W^s=1,W^e,W^z\right]P\left[W^e\right]P\left[W^z|Z\right]}. \label{TransProbConda1Ws1}
	\end{align}
\end{subequations}
The summation \eqref{TransProbConda1Ws0} is the same as \eqref{TransProbConda0}, and is equal to $1$ if the conditions \eqref{a0NonZeroConds} are satisfied; except that condition \eqref{NextEa0} is replaced by $\tilde{E}=E+W^e-1$. In addition, according to \eqref{EvolutionProbZ} - \eqref{EvolutionProbDz}, $P\left[\tilde{s}|s,a=1,W^s=1,W^e,W^z\right]$ in \eqref{TransProbConda1Ws1} will be equal to $1$ for those $\tilde{s}\in S$ that hold all of the following conditions; it will be $0$ otherwise.
\begin{subequations}
	\begin{gather}
		\tilde{Z}=W^z, \\ 
		\tilde{Z}^d=Z, \\ 
		\tilde{E}=E+W^e-1,\\ 
		\left(\tilde{\Delta}^0=0,Z \neq 0\right) \text{ or } \left(\tilde{\Delta}^0=1,Z=0\right)\!, \\
		\left(\tilde{\Delta}^1=0,Z \neq 1\right) \text{ or } \left(\tilde{\Delta}^1=1,Z=1\right)\!.
	\end{gather}
\end{subequations}
Now, we can write $\Delta V(s)$ using the equation \eqref{DV_MainExp} for different values of $Z$ and $Z^d$. 

\textit{\textbf{Case 1.}} $Z=0$ and $Z^d=0$.
In this case, we have:
\begin{align}
	P\left[\tilde{s}|s,a=0,W^s=0,W^e,W^z\right]=
	\begin{cases}
		1 & \parbox[c]{7cm}{$\tilde{Z}=W^z, \tilde{Z}^d=Z^d, \tilde{E}=E+W^e,\\ \tilde{\Delta}^0=\Delta^0+1,\tilde{\Delta}^1=0,$}\\
		0 & \text{otherwise,}
	\end{cases}
\end{align}
\begin{align}
	P\left[\tilde{s}|s,a=1,W^s=0,W^e,W^z\right]=
	\begin{cases}
		1 & \parbox[c]{7cm}{$\tilde{Z}=W^z, \tilde{Z}^d=Z^d, \tilde{E}=E+W^e-1,\\ \tilde{\Delta}^0=\Delta^0+1,\tilde{\Delta}^1=0,$}\\
		0 & \text{otherwise,}
	\end{cases}
\end{align}
\begin{align}
	P\left[\tilde{s}|s,a=1,W^s=1,W^e,W^z\right]=
	\begin{cases}
		1 & \parbox[c]{7cm}{$\tilde{Z}=W^z, \tilde{Z}^d=Z, \tilde{E}=E+W^e-1,\\ \tilde{\Delta}^0=1,\tilde{\Delta}^1=0,$}\\
		0 & \text{otherwise,}
	\end{cases}
\end{align}
then we have $\Delta V(s)=V^1(s)-V^0(s)$, where:
\begin{align}
	V^0(s)&=\sum_{\tilde{s}\in S}{P(\tilde{s}|s,a=0)V(\tilde{s})}  =\sum_{\tilde{s}\in S}{\sum_{[W^e,W^z]}{P\left[\tilde{s}|s,a=0,W^s=0,W^e,W^z\right]P\left[W^e\right]P\left[W^z|Z\right]}V(\tilde{s})} \notag \\ &=\sum_{[W^e,W^z]}{V\left(W^z,Z^d,E+W^e,\Delta^0+1,0\right)P\left[W^e\right]P\left[W^z|Z\right]},
\end{align}
\begin{align}
	V^1(s)&=\sum_{\tilde{s}\in S}{P(\tilde{s}|s,a=1)V(\tilde{s})} \notag \\ &=(1-P_s)\sum_{\tilde{s}\in S}{\sum_{[W^e,W^z]}{P\left[\tilde{s}|s,a=1,W^s=0,W^e,W^z\right]P\left[W^e\right]P\left[W^z|Z\right]}V(\tilde{s})} \notag \\
	&+P_s\sum_{\tilde{s}\in S}{\sum_{[W^e,W^z]}{P\left[\tilde{s}|s,a=1,W^s=1,W^e,W^z\right]P\left[W^e\right]P\left[W^z|Z\right]}V(\tilde{s})} \notag \\ &=(1-P_s)\sum_{[W^e,W^z]}{V\left(W^z,Z^d,E+W^e-1,\Delta^0+1,0\right)P\left[W^e\right]P\left[W^z|Z\right]} \notag \\
	&+P_s\sum_{[W^e,W^z]}{V\left(W^z,Z,E+W^e-1,1,0\right)P\left[W^e\right]P\left[W^z|Z\right]}.
\end{align}

\textit{\textbf{Case 2.}} $Z=0$ and $Z^d=1$.
In this case, we have:
\begin{align}
	P\left[\tilde{s}|s,a=0,W^s=0,W^e,W^z\right]=
	\begin{cases}
		1 & \parbox[c]{7cm}{$\tilde{Z}=W^z, \tilde{Z}^d=Z^d, \tilde{E}=E+W^e,\\ \tilde{\Delta}^0=\Delta^0+1,\tilde{\Delta}^1=\Delta^1+1,$}\\
		0 & \text{otherwise,}
	\end{cases}
\end{align}
\begin{align}
	P\left[\tilde{s}|s,a=1,W^s=0,W^e,W^z\right]=
	\begin{cases}
		1 & \parbox[c]{7cm}{$\tilde{Z}=W^z, \tilde{Z}^d=Z^d, \tilde{E}=E+W^e-1,\\ \tilde{\Delta}^0=\Delta^0+1,\tilde{\Delta}^1=\Delta^1+1,$}\\
		0 & \text{otherwise,}
	\end{cases}
\end{align}
\begin{align}
	P\left[\tilde{s}|s,a=1,W^s=1,W^e,W^z\right]=
	\begin{cases}
		1 & \parbox[c]{7cm}{$\tilde{Z}=W^z, \tilde{Z}^d=Z, \tilde{E}=E+W^e-1,\\ \tilde{\Delta}^0=1,\tilde{\Delta}^1=0,$}\\
		0 & \text{otherwise,}
	\end{cases}
\end{align}
then we have $\Delta V(s)=V^1(s)-V^0(s)$, where:
\begin{align}
	V^0(s)&=\sum_{\tilde{s}\in S}{P(\tilde{s}|s,a=0)V(\tilde{s})} =\sum_{[W^e,W^z]}{V\left(W^z,Z^d,E+W^e,\Delta^0+1,\Delta^1+1\right)P\left[W^e\right]P\left[W^z|Z\right]},
\end{align}
\begin{align}
	V^1(s)&=\sum_{\tilde{s}\in S}{P(\tilde{s}|s,a=1)V(\tilde{s})} \notag \\ &=(1-P_s)\sum_{[W^e,W^z]}{V\left(W^z,Z^d,E+W^e-1,\Delta^0+1,\Delta^1+1\right)P\left[W^e\right]P\left[W^z|Z\right]} \notag \\
	&+P_s\sum_{[W^e,W^z]}{V\left(W^z,Z,E+W^e-1,1,0\right)P\left[W^e\right]P\left[W^z|Z\right]}.
\end{align}

\textit{\textbf{Case 3.}} $Z=1$ and $Z^d=0$.
In this case, we have:
\begin{align}
	P\left[\tilde{s}|s,a=0,W^s=0,W^e,W^z\right]=
	\begin{cases}
		1 & \parbox[c]{7cm}{$\tilde{Z}=W^z, \tilde{Z}^d=Z^d, \tilde{E}=E+W^e,\\ \tilde{\Delta}^0=\Delta^0+1,\tilde{\Delta}^1=\Delta^1+1,$}\\
		0 & \text{otherwise,}
	\end{cases}
\end{align}
\begin{align}
	P\left[\tilde{s}|s,a=1,W^s=0,W^e,W^z\right]=
	\begin{cases}
		1 & \parbox[c]{7cm}{$\tilde{Z}=W^z, \tilde{Z}^d=Z^d, \tilde{E}=E+W^e-1,\\ \tilde{\Delta}^0=\Delta^0+1,\tilde{\Delta}^1=\Delta^1+1,$}\\
		0 & \text{otherwise,}
	\end{cases}
\end{align}
\begin{align}
	P\left[\tilde{s}|s,a=1,W^s=1,W^e,W^z\right]=
	\begin{cases}
		1 & \parbox[c]{7cm}{$\tilde{Z}=W^z, \tilde{Z}^d=Z, \tilde{E}=E+W^e-1,\\ \tilde{\Delta}^0=0,\tilde{\Delta}^1=1,$}\\
		0 & \text{otherwise,}
	\end{cases}
\end{align}
then we have $\Delta V(s)=V^1(s)-V^0(s)$, where:
\begin{align}
	V^0(s)&=\sum_{\tilde{s}\in S}{P(\tilde{s}|s,a=0)V(\tilde{s})} \notag \\ &=\sum_{[W^e,W^z]}{V\left(W^z,Z^d,E+W^e,\Delta^0+1,\Delta^1+1\right)P\left[W^e\right]P\left[W^z|Z\right]},
\end{align}
\begin{align}
	V^1(s)&=\sum_{\tilde{s}\in S}{P(\tilde{s}|s,a=1)V(\tilde{s})} \notag \\ &=(1-P_s)\sum_{[W^e,W^z]}{V\left(W^z,Z^d,E+W^e-1,\Delta^0+1,\Delta^1+1\right)P\left[W^e\right]P\left[W^z|Z\right]} \notag \\
	&+P_s\sum_{[W^e,W^z]}{V\left(W^z,Z,E+W^e-1,0,1\right)P\left[W^e\right]P\left[W^z|Z\right]}.
\end{align}

\textit{\textbf{Case 4.}} $Z=1$ and $Z^d=1$.
In this case, we have:
\begin{align}
	P\left[\tilde{s}|s,a=0,W^s=0,W^e,W^z\right]=
	\begin{cases}
		1 & \parbox[c]{7cm}{$\tilde{Z}=W^z, \tilde{Z}^d=Z^d, \tilde{E}=E+W^e,\\ \tilde{\Delta}^0=0,\tilde{\Delta}^1=\Delta^1+1,$}\\
		0 & \text{otherwise,}
	\end{cases}
\end{align}
\begin{align}
	P\left[\tilde{s}|s,a=1,W^s=0,W^e,W^z\right]=
	\begin{cases}
		1 & \parbox[c]{7cm}{$\tilde{Z}=W^z, \tilde{Z}^d=Z^d, \tilde{E}=E+W^e-1,\\ \tilde{\Delta}^0=0,\tilde{\Delta}^1=\Delta^1+1,$}\\
		0 & \text{otherwise,}
	\end{cases}
\end{align}
\begin{align}
	P\left[\tilde{s}|s,a=1,W^s=1,W^e,W^z\right]=
	\begin{cases}
		1 & \parbox[c]{7cm}{$\tilde{Z}=W^z, \tilde{Z}^d=Z, \tilde{E}=E+W^e-1,\\ \tilde{\Delta}^0=0,\tilde{\Delta}^1=1,$}\\
		0 & \text{otherwise,}
	\end{cases}
\end{align}
then we have $\Delta V(s)=V^1(s)-V^0(s)$, where:
\begin{align}
	V^0(s)&=\sum_{\tilde{s}\in S}{P(\tilde{s}|s,a=0)V(\tilde{s})}  =\sum_{[W^e,W^z]}{V\left(W^z,Z^d,E+W^e,0,\Delta^1+1\right)P\left[W^e\right]P\left[W^z|Z\right]},
\end{align}
\begin{align}
	V^1(s)&=\sum_{\tilde{s}\in S}{P(\tilde{s}|s,a=1)V(\tilde{s})} =(1-P_s)\sum_{[W^e,W^z]}{V\left(W^z,Z^d,E+W^e-1,0,\Delta^1+1\right)P\left[W^e\right]P\left[W^z|Z\right]} \notag \\
	&+P_s\sum_{[W^e,W^z]}{V\left(W^z,Z,E+W^e-1,0,1\right)P\left[W^e\right]P\left[W^z|Z\right]}.
\end{align}
We will proceed with the theorem's proof for case 3, as the proof for other cases follows a similar approach. We aim to show that $\Delta V(s)$ is decreasing in $(\Delta^0,\Delta^1)$ for each $(Z,Z^d,E)$. 
\begin{align}
	\label{DeltaVcase3}
	\Delta V(s)&=V^1(s)-V^0(s)\\
	&=\sum_{[W^e,W^z]}\bigg\{(1-P_s)V\left(W^z,Z^d,E+W^e-1,\Delta^0+1,\Delta^1+1\right) \notag \\
	&\phantom{=\sum_{[W^e,W^z]}\bigg\{}-V\left(W^z,Z^d,E+W^e,\Delta^0+1,\Delta^1+1\right) \notag \\
	&\phantom{=\sum_{[W^e,W^z]}\bigg\{} +P_s V\left(W^z,Z,E+W^e-1,0,1\right)\bigg\}P\left[W^e\right]P\left[W^z|Z\right].
\end{align}
Let us define $s^+=\left[Z,Z^d,E,\Delta^{0+},\Delta^{1+}\right]^T$ and $s^-=\left[Z,Z^d,E,\Delta^{0-},\Delta^{1-}\right]^T$ such that $\left(\Delta^{0+},\Delta^{1+}\right) \geq \left(\Delta^{0-},\Delta^{1-}\right)$, element-wise. We will therefore prove that $\Delta V(s^+) \leq \Delta V(s^-)$.

{\small 
	\begin{align}
		\Delta V(s^+) \leq \Delta V(s^-) &\Leftrightarrow \sum_{[W^e,W^z]}\bigg\{(1-P_s)V\left(W^z,Z^d,E+W^e-1,\Delta^{0+}+1,\Delta^{1+}+1\right) \notag \\
		&\phantom{=\sum_{[W^e,W^z]}\bigg\{}-V\left(W^z,Z^d,E+W^e,\Delta^{0+}+1,\Delta^{1+}+1\right) \notag \\
		&\phantom{=\sum_{[W^e,W^z]}\bigg\{} +P_s V\left(W^z,Z,E+W^e-1,0,1\right)\bigg\}P\left[W^e\right]P\left[W^z|Z\right] \notag \\
		&\leq \sum_{[W^e,W^z]}\bigg\{(1-P_s)V\left(W^z,Z^d,E+W^e-1,\Delta^{0-}+1,\Delta^{1-}+1\right) \notag \\
		&\phantom{=\sum_{[W^e,W^z]}\bigg\{}-V\left(W^z,Z^d,E+W^e,\Delta^{0-}+1,\Delta^{1-}+1\right) \notag \\
		&\phantom{=\sum_{[W^e,W^z]}\bigg\{} +P_s V\left(W^z,Z,E+W^e-1,0,1\right)\bigg\}P\left[W^e\right]P\left[W^z|Z\right] \notag 
	\end{align}
	\begin{align}
		&\Leftrightarrow (1-P_s)V\left(W^z,Z^d,E+W^e-1,\Delta^{0+}+1,\Delta^{1+}+1\right)-V\left(W^z,Z^d,E+W^e,\Delta^{0+}+1,\Delta^{1+}+1\right) \notag \\
		&\phantom{\Leftrightarrow} \leq (1-P_s)V\left(W^z,Z^d,E+W^e-1,\Delta^{0-}+1,\Delta^{1-}+1\right)-V\left(W^z,Z^d,E+W^e,\Delta^{0-}+1,\Delta^{1-}+1\right) \notag \\
		&\Leftrightarrow (1-P_s)\left[V\left(W^z,Z^d,E+W^e-1,\Delta^{0+}+1,\Delta^{1+}+1\right) -V\left(W^z,Z^d,E+W^e-1,\Delta^{0-}+1,\Delta^{1-}+1\right)\right] \notag \\
		&\phantom{\Leftrightarrow}\leq V\left(W^z,Z^d,E+W^e,\Delta^{0+}+1,\Delta^{1+}+1\right) -V\left(W^z,Z^d,E+W^e,\Delta^{0-}+1,\Delta^{1-}+1\right).
	\end{align}
}

In the following Lemma, we demonstrate the last inequality which concludes the proof of Theorem 1.
\end{proof}

\section{Lemma 1}
\label{Lemma1}
\begin{lemma}
    Suppose that $s_E^+=[Z,Z^d,E,\Delta^{0+},\Delta^{0+}]^T$, $s_E^-=[Z,Z^d,E,\Delta^{0-},\Delta^{0-}]^T$, $s_{E-1}^+=[Z,Z^d,E-1,\Delta^{0+},\Delta^{0+}]^T$, and $s_{E-1}^-=[Z,Z^d,E-1,\Delta^{0-},\Delta^{0-}]^T$ are four states such that $E\in\{1,2,3,\cdots\}$ and $\left(\Delta^{0+},\Delta^{1+}\right) \geq \left(\Delta^{0-},\Delta^{1-}\right)$; then the value function satisfies the following inequality:
    \begin{align}
    	(1-P_s)\left[V\left(s_{E-1}^+\right) -V\left(s_{E-1}^-\right)\right] \leq V\left(s_{E}^+\right) -V\left(s_{E}^-\right).
    \end{align}
\end{lemma}

\begin{proof} 
We employ the Value Iteration Algorithm (VIA) to prove the lemma. In each iteration at time step $k$, the value function is updated as follows:

\begin{align}
	\label{VIA}
	V_k(s) &= \min_{a\in\{0,1\}}\left\{\sum_{\tilde{s}\in S}{P(\tilde{s}|s,a)\left[g(s,a)+\gamma V_{k-1}(\tilde{s})\right]}\right\} 
	=g(s)+\min_{a\in\{0,1\}}\left\{\gamma\sum_{\tilde{s}\in S}{P(\tilde{s}|s,a)V_{k-1}(\tilde{s})}\right\}.
\end{align}

VIA converges to the value function of the Bellman's equation irrespective of the initial value assigned to $V_0(s)$, i.e., $\lim_{k\rightarrow\infty}{V_k(s)}=V(s)\ \forall s\in S$. Therefore, it suffices to establish the following:
\begin{align}
	\label{Lemma1_iter_k}
	(1-P_s)\left[V_k\left(s_{E-1}^+\right) -V_k\left(s_{E-1}^-\right)\right] \leq V_k\left(s_{E}^+\right) -V_k\left(s_{E}^-\right), \quad \forall k=0,1,2,\cdots.
\end{align}

We utilize mathematical induction to proceed the proof. Assuming $V_0(s) = 0$ for all $s \in S$, \eqref{Lemma1_iter_k} holds true for $k=0$. Now, with the same assumption extending up to $k > 0$, we prove its validity for $k+1$, i.e.:
\begin{align}
	\label{Lemma1_iter_k+1}
	&(1-P_s)\left[V_{k+1}\left(s_{E-1}^+\right) -V_{k+1}\left(s_{E-1}^-\right)\right] \leq V_{k+1}\left(s_{E}^+\right) -V_{k+1}\left(s_{E}^-\right) \notag \\
	&\Leftrightarrow (1-P_s)\left[V_{k+1}\left(s_{E-1}^+\right) -V_{k+1}\left(s_{E-1}^-\right)\right]-\left[V_{k+1}\left(s_{E}^+\right) -V_{k+1}\left(s_{E}^-\right)\right] \leq 0.
\end{align}
Let us define $V_{k+1}^0(s)$ and $V_{k+1}^1(s)$ as follows:
\begin{align}
	\label{V_k+1_a0}
	V_{k+1}^0(s)&=g(s)+\gamma\sum_{\tilde{s}\in S}{P(\tilde{s}|s,a=0)V_{k}(\tilde{s})} \notag \\
	&=g(s)+\gamma\sum_{[W^e,W^z]}{V_k\left(W^z,Z^d,E+W^e,\Delta^0+1,\Delta^1+1\right)P\left[W^e\right]P\left[W^z|Z\right]},
\end{align}
\begin{align}
	\label{V_k+1_a1}
	V_{k+1}^1(s)&=g(s)+\gamma\sum_{\tilde{s}\in S}{P(\tilde{s}|s,a=1)V_{k}(\tilde{s})} \notag \\
	&=g(s)+\gamma(1-P_s)\sum_{[W^e,W^z]}{V_k\left(W^z,Z^d,E+W^e-1,\Delta^0+1,\Delta^1+1\right)P\left[W^e\right]P\left[W^z|Z\right]} \notag \\
	&+\gamma P_s\sum_{[W^e,W^z]}{V_k\left(W^z,Z,E+W^e-1,0,1\right)P\left[W^e\right]P\left[W^z|Z\right]},
\end{align}
\noindent then we have $V_{k+1}(s)=\min\left\{V_{k+1}^0(s),V_{k+1}^1(s)\right\}$, according to VIA iteration \eqref{VIA} at time slot $k+1$. Thus, \eqref{Lemma1_iter_k+1} can be rewritten as follows:
\begin{align}
	\label{Lemma1_iter_k+1_min}
	(1-P_s)&\left[\min\left\{V_{k+1}^0(s_{E-1}^+),V_{k+1}^1(s_{E-1}^+)\right\} -\min\left\{V_{k+1}^0(s_{E-1}^-),V_{k+1}^1(s_{E-1}^-)\right\}\right] \notag \\
	& -\left[\min\left\{V_{k+1}^0(s_{E}^+),V_{k+1}^1(s_{E}^+)\right\} -\min\left\{V_{k+1}^0(s_{E}^-),V_{k+1}^1(s_{E}^-)\right\}\right] \leq 0.
\end{align}
Now, we consider four cases.

\textit{\textbf{Case 1.}} $V_{k+1}^0(s_{E-1}^-) \leq V_{k+1}^1(s_{E-1}^-)$ and $V_{k+1}^0(s_{E}^+) \leq V_{k+1}^1(s_{E}^+)$.
In this case, equation \eqref{Lemma1_iter_k+1_min} is simplified to:

\begin{align}
	(1-P_s)&\left[\min\left\{V_{k+1}^0(s_{E-1}^+),V_{k+1}^1(s_{E-1}^+)\right\} -V_{k+1}^0(s_{E-1}^-)\right] \notag \\
	& -\left[V_{k+1}^0(s_{E}^+) -\min\left\{V_{k+1}^0(s_{E}^-),V_{k+1}^1(s_{E}^-)\right\}\right] \leq 0.
\end{align}
We know that $\min\left\{x,y\right\}=x+\min\left\{0,y-x\right\}$, so we can simplify further:
\begin{align}
	(1-P_s)&\left[V_{k+1}^0(s_{E-1}^+)-V_{k+1}^0(s_{E-1}^-)\right]+\overbrace{(1-P_s)\min\left\{0,V_{k+1}^1(s_{E-1}^+)-V_{k+1}^0(s_{E-1}^+)\right\}}^{\leq 0}  \notag \\
	& -\left[V_{k+1}^0(s_{E}^+) -V_{k+1}^0(s_{E}^-)\right]+\underbrace{\min\left\{0,V_{k+1}^1(s_{E}^-)-V_{k+1}^0(s_{E}^-)\right\}}_{\leq 0} \leq 0,
\end{align}
\noindent where the second and last terms are negative (non-positive), thus it suffices to show that:
\begin{align}
	(1-P_s)&\left[V_{k+1}^0(s_{E-1}^+)-V_{k+1}^0(s_{E-1}^-)\right]-\left[V_{k+1}^0(s_{E}^+) -V_{k+1}^0(s_{E}^-)\right] \leq 0.
\end{align}
According to \eqref{V_k+1_a0}, we have:
\begin{align}
	\label{Lemma1_Case1_Ineq}
	&(1-P_s)\left[g(s_{E-1}^+)-g(s_{E-1}^-)+\gamma\sum_{\tilde{s}\in S}{\left[P(\tilde{s}|s_{E-1}^+,a=0)-P(\tilde{s}|s_{E-1}^-,a=0)\right]V_{k}(\tilde{s})}\right] \notag \\ &\hphantom{(1-P_s)}-\left[g(s_{E}^+)-g(s_{E}^-)+\gamma\sum_{\tilde{s}\in S}{\left[P(\tilde{s}|s_{E}^+,a=0)-P(\tilde{s}|s_{E}^-,a=0)\right]V_{k}(\tilde{s})}\right] \leq 0 \notag \\
	&\Leftrightarrow (1-P_s)\left[g(s_{E-1}^+)-g(s_{E-1}^-)\right]-\left[g(s_{E}^+)-g(s_{E}^-)\right] \notag \\ 
	&\hphantom{(1-P_s)}+\gamma\sum_{\tilde{s}\in S}\Big\{(1-P_s)\left[P(\tilde{s}|s_{E-1}^+,a=0)-P(\tilde{s}|s_{E-1}^-,a=0)\right] \notag \\
	&\hphantom{(1-P_s)+\gamma\sum_{\tilde{s}\in S}\ }-\left[P(\tilde{s}|s_{E}^+,a=0)-P(\tilde{s}|s_{E}^-,a=0)\right]\Big\}V_{k}(\tilde{s}) \leq 0.
\end{align}

We know that $g(s_{E-1}^+)=g(s_{E}^+)=(1-Z)\Delta^{0+}+Z(\Delta^{1+})^2$, $g(s_{E}^-)=g(s_{E-1}^-)=(1-Z)\Delta^{0-}+Z(\Delta^{1-})^2$. Additionally, we have $g(s_{E-1}^+)=g(s_{E}^+) \geq g(s_{E-1}^-)=g(s_{E}^-)$ since $\Delta^{0+} \geq \Delta^{0-}$ and $\Delta^{1+} \geq \Delta^{1-}$. Therefore, $(1-P_s)\left[g(s_{E-1}^+)-g(s_{E-1}^-)\right]-\left[g(s_{E}^+)-g(s_{E}^-)\right]=-P_s\left[g(s_{E}^+)-g(s_{E}^-)\right] \leq 0$. Now, we prove that the summation in \eqref{Lemma1_Case1_Ineq} is also negative. Simplifying this summation based on equation \eqref{V_k+1_a0} results in the following expression:

{\footnotesize 
	\begin{align}
		\label{Lemma1_Case1_Summ}
		\sum_{[W^e,W^z]}&\Big\{(1-P_s)\left[V_k\left(W^z,Z^d,E+W^e-1,\Delta^{0+}+1,\Delta^{1+}+1\right)-V_k\left(W^z,Z^d,E+W^e-1,\Delta^{0-}+1,\Delta^{1-}+1\right)\right] \notag \\
		&-\left[V_k\left(W^z,Z^d,E+W^e,\Delta^{0+}+1,\Delta^{1+}+1\right)-V_k\left(W^z,Z^d,E+W^e,\Delta^{0-}+1,\Delta^{1-}+1\right)\right]\Big\}P\left[W^e\right]P\left[W^z|Z\right] \leq 0.
	\end{align}
}

Let us define $\tilde{s}_E^+=\tilde{s}_E^+(W^e,W^z)=[W^z,Z^d,E+W^e,\Delta^{0+}+1,\Delta^{1+}+1]^T$ and $\tilde{s}_E^+=\tilde{s}_E^-(W^e,W^z)=[W^z,Z^d,E+W^e,\Delta^{0-}+1,\Delta^{1-}+1]^T$, then \eqref{Lemma1_Case1_Summ} can be rewritten:
\begin{align}
	\sum_{[W^e,W^z]}\left\{(1-P_s)\left[V_k\left(\tilde{s}_{E-1}^+\right)-V_k\left(\tilde{s}_{E-1}^-\right)\right] -\left[V_k\left(\tilde{s}_{E}^+\right)-V_k\left(\tilde{s}_{E}^-\right)\right]\right\}P\left[W^e\right]P\left[W^z|Z\right] \leq 0.
\end{align}
In accordance with the assumption stated in equation \eqref{Lemma1_iter_k}, $(1-P_s)\left[V_k\left(\tilde{s}_{E-1}^+\right)-V_k\left(\tilde{s}_{E-1}^-\right)\right] -\left[V_k\left(\tilde{s}_{E}^+\right)-V_k\left(\tilde{s}_{E}^-\right)\right] \leq 0$. As a result, the summation \eqref{Lemma1_Case1_Summ} is negative (non-positive). This observation concludes the proof for case 1.

\textit{\textbf{Case 2.}} $V_{k+1}^0(s_{E-1}^-) \leq V_{k+1}^1(s_{E-1}^-)$ and $V_{k+1}^0(s_{E}^+) > V_{k+1}^1(s_{E}^+)$.
In this case, equation \eqref{Lemma1_iter_k+1_min} is reduced to:
\begin{align}
	(1-P_s)&\left[\min\left\{V_{k+1}^0(s_{E-1}^+),V_{k+1}^1(s_{E-1}^+)\right\} -V_{k+1}^0(s_{E-1}^-)\right] \notag \\
	& -\left[V_{k+1}^1(s_{E}^+) -\min\left\{V_{k+1}^0(s_{E}^-),V_{k+1}^1(s_{E}^-)\right\}\right] \leq 0 \notag \\
	& \Leftrightarrow (1-P_s)\left[V_{k+1}^0(s_{E-1}^+)-V_{k+1}^0(s_{E-1}^-)\right]+\overbrace{(1-P_s)\min\left\{0,V_{k+1}^1(s_{E-1}^+)-V_{k+1}^0(s_{E-1}^+)\right\}}^{\leq 0}  \notag \\
	& -\left[V_{k+1}^1(s_{E}^+) -V_{k+1}^1(s_{E}^-)\right]+\underbrace{\min\left\{V_{k+1}^0(s_{E}^-)-V_{k+1}^1(s_{E}^-),0\right\}}_{\leq 0} \leq 0.
\end{align}

It is adequate to show that:
\begin{align}
	(1-P_s)&\left[V_{k+1}^0(s_{E-1}^+)-V_{k+1}^0(s_{E-1}^-)\right]-\left[V_{k+1}^1(s_{E}^+) -V_{k+1}^1(s_{E}^-)\right] \leq 0.
\end{align}
According to \eqref{V_k+1_a0} and \eqref{V_k+1_a1}, we have:
\begin{align}
	\label{Lemma1_Case2_Ineq}
	&(1-P_s)\left[g(s_{E-1}^+)-g(s_{E-1}^-)+\gamma\sum_{\tilde{s}\in S}{\left[P(\tilde{s}|s_{E-1}^+,a=0)-P(\tilde{s}|s_{E-1}^-,a=0)\right]V_{k}(\tilde{s})}\right] \notag \\ &\hphantom{(1-P_s)}-\left[g(s_{E}^+)-g(s_{E}^-)+\gamma\sum_{\tilde{s}\in S}{\left[P(\tilde{s}|s_{E}^+,a=1)-P(\tilde{s}|s_{E}^-,a=1)\right]V_{k}(\tilde{s})}\right] \leq 0 \notag \\
	&\Leftrightarrow \overbrace{(1-P_s)\left[g(s_{E-1}^+)-g(s_{E-1}^-)\right]-\left[g(s_{E}^+)-g(s_{E}^-)\right]}^{=-P_s\left[g(s_{E}^+)-g(s_{E}^-)\right]\leq 0} \notag \\ 
	&\hphantom{(1-P_s)}+\gamma\sum_{\tilde{s}\in S}\Big\{(1-P_s)\left[P(\tilde{s}|s_{E-1}^+,a=0)-P(\tilde{s}|s_{E-1}^-,a=0)\right] \notag \\
	&\hphantom{(1-P_s)+\gamma\sum_{\tilde{s}\in S}\ }-\left[P(\tilde{s}|s_{E}^+,a=1)-P(\tilde{s}|s_{E}^-,a=1)\right]\Big\}V_{k}(\tilde{s}) \leq 0.
\end{align}

We prove that the summation in \eqref{Lemma1_Case2_Ineq} is also non-positive. Rewriting this sum using the equations \eqref{V_k+1_a0} and \eqref{V_k+1_a1} yields the subsequent equation:
\begin{align}
	\sum_{[W^e,W^z]}(1-P_s)\Big\{\left[V_k\left(\tilde{s}_{E-1}^+\right)-V_k\left(\tilde{s}_{E-1}^-\right)\right] -\left[V_k\left(\tilde{s}_{E-1}^+\right)-V_k\left(\tilde{s}_{E-1}^-\right)\right]\Big\}P\left[W^e\right]P\left[W^z|Z\right] = 0,
\end{align}
\noindent and it concludes the proof for case 2.

\textit{\textbf{Case 3.}} $V_{k+1}^0(s_{E-1}^-) > V_{k+1}^1(s_{E-1}^-)$ and $V_{k+1}^0(s_{E}^+) \leq V_{k+1}^1(s_{E}^+)$\footnote{It is noteworthy that this case does not occur when $E=1$, as the action $a=0$ is optimal, and $V_{k+1}^0(s_{E-1}^-) = V_{k+1}^1(s_{E-1}^-)$.}.
In this case, equation \eqref{Lemma1_iter_k+1_min} is simplified to:
\begin{align}
	(1-P_s)&\left[\min\left\{V_{k+1}^0(s_{E-1}^+),V_{k+1}^1(s_{E-1}^+)\right\} -V_{k+1}^1(s_{E-1}^-)\right] \notag \\
	& -\left[V_{k+1}^0(s_{E}^+) -\min\left\{V_{k+1}^0(s_{E}^-),V_{k+1}^1(s_{E}^-)\right\}\right] \leq 0 \notag \\
	& \Leftrightarrow (1-P_s)\left[V_{k+1}^1(s_{E-1}^+)-V_{k+1}^1(s_{E-1}^-)\right]+\overbrace{(1-P_s)\min\left\{V_{k+1}^0(s_{E-1}^+)-V_{k+1}^1(s_{E-1}^+),0\right\}}^{\leq 0}  \notag \\
	& -\left[V_{k+1}^0(s_{E}^+) -V_{k+1}^0(s_{E}^-)\right]+\underbrace{\min\left\{0,V_{k+1}^1(s_{E}^-)-V_{k+1}^0(s_{E}^-)\right\}}_{\leq 0} \leq 0,
\end{align}
It is adequate to demonstrate that:
\begin{align}
	(1-P_s)&\left[V_{k+1}^1(s_{E-1}^+)-V_{k+1}^1(s_{E-1}^-)\right]-\left[V_{k+1}^0(s_{E}^+) -V_{k+1}^0(s_{E}^-)\right] \leq 0.
\end{align}
According to \eqref{V_k+1_a0} and \eqref{V_k+1_a1}:
\begin{align}
	\label{Lemma1_Case3_Ineq}
	&(1-P_s)\left[g(s_{E-1}^+)-g(s_{E-1}^-)+\gamma\sum_{\tilde{s}\in S}{\left[P(\tilde{s}|s_{E-1}^+,a=1)-P(\tilde{s}|s_{E-1}^-,a=1)\right]V_{k}(\tilde{s})}\right] \notag \\ &\hphantom{(1-P_s)}-\left[g(s_{E}^+)-g(s_{E}^-)+\gamma\sum_{\tilde{s}\in S}{\left[P(\tilde{s}|s_{E}^+,a=0)-P(\tilde{s}|s_{E}^-,a=0)\right]V_{k}(\tilde{s})}\right] \leq 0 \notag \\
	&\Leftrightarrow \overbrace{(1-P_s)\left[g(s_{E-1}^+)-g(s_{E-1}^-)\right]-\left[g(s_{E}^+)-g(s_{E}^-)\right]}^{=-P_s\left[g(s_{E}^+)-g(s_{E}^-)\right]\leq 0} \notag \\ 
	&\hphantom{(1-P_s)}+\gamma\sum_{\tilde{s}\in S}\Big\{(1-P_s)\left[P(\tilde{s}|s_{E-1}^+,a=1)-P(\tilde{s}|s_{E-1}^-,a=1)\right] \notag \\
	&\hphantom{(1-P_s)+\gamma\sum_{\tilde{s}\in S}\ }-\left[P(\tilde{s}|s_{E}^+,a=0)-P(\tilde{s}|s_{E}^-,a=0)\right]\Big\}V_{k}(\tilde{s}) \leq 0.
\end{align}

We prove that the sum in \eqref{Lemma1_Case3_Ineq} is likewise non-positive. By simplifying this sum using equations \eqref{V_k+1_a0} and \eqref{V_k+1_a1}, we arrive at the subsequent expression:

\begin{align}
	&\sum_{[W^e,W^z]}\Big\{(1-P_s)^2\left[V_k\left(\tilde{s}_{E-2}^+\right)-V_k\left(\tilde{s}_{E-2}^-\right)\right] -\left[V_k\left(\tilde{s}_{E}^+\right)-V_k\left(\tilde{s}_{E}^-\right)\right]\Big\}P\left[W^e\right]P\left[W^z|Z\right] \notag \\
	&=\sum_{[W^e,W^z]}\bigg\{\underbrace{(1-P_s)^2\left[V_k\left(\tilde{s}_{E-2}^+\right)-V_k\left(\tilde{s}_{E-2}^-\right)\right]-(1-P_s)\left[V_k\left(\tilde{s}_{E-1}^+\right)-V_k\left(\tilde{s}_{E-1}^-\right)\right]}_{(a)} \notag \\
	&\hphantom{\sum_{[W^e,W^z]}} + \underbrace{(1-P_s)\left[V_k\left(\tilde{s}_{E-1}^+\right)-V_k\left(\tilde{s}_{E-1}^-\right)\right] -\left[V_k\left(\tilde{s}_{E}^+\right)-V_k\left(\tilde{s}_{E}^-\right)\right]}_{(b)}\bigg\}P\left[W^e\right]P\left[W^z|Z\right], \notag
\end{align}

\noindent where both $(a)$ and $(b)$ are negative according to \eqref{Lemma1_iter_k}, thus concluding the proof for case 3.
\begin{align}
	(a)&=(1-P_s)\Big\{(1-P_s)\left[V_k\left(\tilde{s}_{E-2}^+\right)-V_k\left(\tilde{s}_{E-2}^-\right)\right]-\left[V_k\left(\tilde{s}_{E-1}^+\right)-V_k\left(\tilde{s}_{E-1}^-\right)\right]\Big\} \notag \\
	& \leq (1-P_s)\left[V_k\left(\tilde{s}_{E-2}^+\right)-V_k\left(\tilde{s}_{E-2}^-\right)\right]-\left[V_k\left(\tilde{s}_{E-1}^+\right)-V_k\left(\tilde{s}_{E-1}^-\right)\right] \leq 0.
\end{align}

\textit{\textbf{Case 4.}} $V_{k+1}^0(s_{E-1}^-) > V_{k+1}^1(s_{E-1}^-)$ and $V_{k+1}^0(s_{E}^+) > V_{k+1}^1(s_{E}^+)$\footnote{It is noteworthy that this case does not occur when $E=1$, as the action $a=0$ is optimal, and $V_{k+1}^0(s_{E-1}^-) = V_{k+1}^1(s_{E-1}^-)$.}.
In this case, equation \eqref{Lemma1_iter_k+1_min} is reduced to:
\begin{align}
	(1-P_s)&\left[\min\left\{V_{k+1}^0(s_{E-1}^+),V_{k+1}^1(s_{E-1}^+)\right\} -V_{k+1}^1(s_{E-1}^-)\right] \notag \\
	& -\left[V_{k+1}^1(s_{E}^+) -\min\left\{V_{k+1}^0(s_{E}^-),V_{k+1}^1(s_{E}^-)\right\}\right] \leq 0 \notag \\
	& \Leftrightarrow (1-P_s)\left[V_{k+1}^1(s_{E-1}^+)-V_{k+1}^1(s_{E-1}^-)\right]+\overbrace{(1-P_s)\min\left\{V_{k+1}^0(s_{E-1}^+)-V_{k+1}^1(s_{E-1}^+),0\right\}}^{\leq 0}  \notag \\
	& -\left[V_{k+1}^1(s_{E}^+) -V_{k+1}^1(s_{E}^-)\right]+\underbrace{\min\left\{V_{k+1}^0(s_{E}^-)-V_{k+1}^1(s_{E}^-),0\right\}}_{\leq 0} \leq 0.
\end{align}

It suffices to demonstrate that:
\begin{align}
	(1-P_s)&\left[V_{k+1}^1(s_{E-1}^+)-V_{k+1}^1(s_{E-1}^-)\right]-\left[V_{k+1}^1(s_{E}^+) -V_{k+1}^1(s_{E}^-)\right] \leq 0.
\end{align}

According to \eqref{V_k+1_a1}:

\begin{align}
	\label{Lemma1_Case4_Ineq}
	&(1-P_s)\left[g(s_{E-1}^+)-g(s_{E-1}^-)+\gamma\sum_{\tilde{s}\in S}{\left[P(\tilde{s}|s_{E-1}^+,a=1)-P(\tilde{s}|s_{E-1}^-,a=1)\right]V_{k}(\tilde{s})}\right] \notag \\ &\hphantom{(1-P_s)}-\left[g(s_{E}^+)-g(s_{E}^-)+\gamma\sum_{\tilde{s}\in S}{\left[P(\tilde{s}|s_{E}^+,a=1)-P(\tilde{s}|s_{E}^-,a=1)\right]V_{k}(\tilde{s})}\right] \leq 0 \notag \\
	&\Leftrightarrow \overbrace{(1-P_s)\left[g(s_{E-1}^+)-g(s_{E-1}^-)\right]-\left[g(s_{E}^+)-g(s_{E}^-)\right]}^{=-P_s\left[g(s_{E}^+)-g(s_{E}^-)\right]\leq 0} \notag \\ 
	&\hphantom{(1-P_s)}+\gamma\sum_{\tilde{s}\in S}\Big\{(1-P_s)\left[P(\tilde{s}|s_{E-1}^+,a=1)-P(\tilde{s}|s_{E-1}^-,a=1)\right] \notag \\
	&\hphantom{(1-P_s)+\gamma\sum_{\tilde{s}\in S}\ }-\left[P(\tilde{s}|s_{E}^+,a=1)-P(\tilde{s}|s_{E}^-,a=1)\right]\Big\}V_{k}(\tilde{s}) \leq 0.
\end{align}

We demonstrate that the sum in \eqref{Lemma1_Case4_Ineq} is also non-positive. Simplifying the summation using the equation \eqref{V_k+1_a1} yields the subsequent expression:
{\small
	\begin{align}
		&\sum_{[W^e,W^z]}\Big\{(1-P_s)^2\left[V_k\left(\tilde{s}_{E-2}^+\right)-V_k\left(\tilde{s}_{E-2}^-\right)\right] -(1-P_s)\left[V_k\left(\tilde{s}_{E-1}^+\right)-V_k\left(\tilde{s}_{E-1}^-\right)\right]\Big\}P\left[W^e\right]P\left[W^z|Z\right] \leq 0,
	\end{align}
}
\noindent and the proof for case 4 and Lemma 1 is completed.
\end{proof}

\end{document}